\documentclass[11pt]{article}

\usepackage[colorlinks=true, allcolors=blue]{hyperref}
\usepackage[english]{babel}

\usepackage[letterpaper,top=1in,bottom=1in,left=1in,right=1in]{geometry}

\usepackage[nolist,nohyperlinks]{acronym}

\usepackage{scrextend}
\usepackage[dvipsnames]{xcolor}
\newcommand{\todo}[1]{\noindent\textcolor{Green}{\textsc{todo:} #1}}

\usepackage{amsmath, amscd, amsthm, amsfonts}
\usepackage{amssymb}
\usepackage{mathtools}
\usepackage{thmtools}
\usepackage{hyperref}
\usepackage{subfig}
\usepackage{graphicx}
\usepackage{svg}   
\usepackage{enumerate}  
\usepackage{csquotes} 
\usepackage{booktabs}  
\usepackage{multirow}
\usepackage{braket} 
\usepackage{verbatim}

\usepackage{tikz}

\usetikzlibrary{snakes}
\usetikzlibrary{patterns}
\usetikzlibrary{backgrounds}

\usepackage{dsfont}

\newtheorem{theorem}{Theorem}
\newtheorem{observation}[theorem]{Observation}
\newtheorem{corollary}[theorem]{Corollary}

\newtheorem{lemma}{Lemma}

\newcommand{\zz}{\mathbb{Z}}

\newcommand{\qq}{\mathbb{Q}}
\newcommand{\fract}{\mathrm{frac}}
\newcommand{\dd}{d}
\newcommand{\notd}{t}

\DeclarePairedDelimiter{\nint}\lfloor\rceil
\DeclarePairedDelimiter{\floor}\lfloor\rfloor
\newcommand{\norm}[1]{\left\lVert #1 \right\rVert}
\newcommand{\abs}[1]{\left\lvert #1 \right\rvert}

\DeclareMathOperator{\lattice}{\mathcal{L}}

\DeclareMathOperator{\OTilde}{\ensuremath{\tilde{O}}}

\DeclareMathOperator{\rem}{rem}

\newcommand{\ie}{i.\,e.}

\newcommand{\eg}{e.\,g.}

\usepackage{algorithmicx}
\usepackage[noend]{algpseudocode}
\usepackage{algorithm}
\newcommand*\Let[2]{\State #1 $\gets$ #2}

\algrenewcommand\alglinenumber[1]{\sf\footnotesize \texttt{#1\,}}%

\algrenewcommand\algorithmicrequire{\textbf{Precondition:}}
\algrenewcommand\algorithmicensure{\textbf{Postcondition:}}

\usepackage{multicol}
\usepackage{tcolorbox}

\title{Faster Lattice Basis Computation via a Natural Generalization\\ of the Euclidean Algorithm
}
\author{
	Kim-Manuel Klein
	\\\small University of Lübeck
	\\\small kimmanuel.klein@uni-luebeck.de
	\and
	Janina Reuter
	\\\small Kiel University
	\\\small janina.reuter@informatik.uni-kiel.de
}
\date{}

\begin{document}
\maketitle
\begin{abstract}
The Euclidean algorithm is one of the oldest algorithms known to mankind. Given two integral numbers $a_1$ and $a_2$, it computes the greatest common divisor (gcd) of $a_1$ and $a_2$ in a very elegant way. From a lattice perspective, it computes a basis of the lattice generated by $a_1$ and $a_2$ as $\gcd(a_1,a_2) \mathbb{Z} = a_1 \mathbb{Z} + a_2 \mathbb{Z}$. In this paper, we show that the classical Euclidean algorithm can be adapted in a very natural way to compute a basis of a lattice that is generated by vectors $A_1, \ldots , A_n \in \mathbb{Z}^d$ with $n> \mathrm{rank}(A_1, \ldots ,A_n)$. 
Similar to the Euclidean algorithm, our algorithm is easy to describe and implement and can be written within 12 lines of pseudocode.

As our main result, we obtain an algorithm to compute a lattice basis for given vectors $A_1, \ldots , A_n \in \mathbb{Z}^d$ in time (counting bit operations) $LS + \tilde{O}((n-d)d^2 \cdot \log(||A||)$, where $LS$ is the time required to obtain the exact fractional solution of a certain system of linear equalities. The analysis of the running time of our algorithms relies on fundamental statements on the fractionality of solutions of linear systems of equations. 

So far, the fastest algorithm for lattice basis computation was due to Storjohann and Labahn (ISSAC 1996) having a running time of $\tilde{O}(nd^\omega\log ||A||)$, where $\omega$ denotes the matrix multiplication exponent. We can improve upon their running time as our algorithm requires at most $\tilde{O}(\max\{n-d, d^2\}d^{\omega(2)-1} \log ||A||)$ bit operations, where $\omega(2)$ denotes the exponent for multiplying a $n\times n$ matrix with an $n\times n^2$ matrix. For current values of $\omega$ and $\omega(2)$, our algorithm improves the running time therefore by a factor of at least $d^{0.12}$ (since $n>d$) providing the first general runtime improvement for lattice basis computation in nearly 30 years. In the cases of either few additional vectors, e.g. $n-d \in d^{o(1)}$, or a very large number of additional vectors, e.g. $n-d \in \Omega (d^k)$ and $k>1$, the run time improves even further in comparison.

At last, we present a postprocessing procedure which yields an improved size bound of $\sqrt{d} ||A||$ for vectors of the resulting basis matrix. The procedure only requires $\tilde{O}(d^3 \log ||A||)$ bit operations. By this we improve upon the running time of previous results by a factor of at least $d^{0.74}$.
\end{abstract}


\section{Introduction}\label{sec:intro}
Given two integral numbers $a_1$ and $a_2$, the Euclidean algorithm computes the greatest common divisor (gcd) of $a_1$ and $a_2$ in a very elegant way. Starting with $s = a_1$ and $t= a_2$, a residue $r$ is being computed by setting
\begin{align*}
    r = \min_{x \in \zz} \{ r \in \zz \mid s x + r = t \} = \min \{ t \pmod s, |(t \pmod s) - s| \}.
\end{align*}
This procedure is continued iteratively with $s = t$ and $t = r$ until $r$ equals $0$. Since $r \leq \lfloor t/2 \rfloor$ the algorithm terminates after at most $\log(\min \{ a_1, a_2\})$ many iterations.

An alternative interpretation of the gcd or the Euclidean algorithm is the following: Consider all integers that are divisible by $a_1$ or respectively $a_2$, which is the set $a_1 \zz$ or respectively the set $a_2 \zz$. Consider their sum (i.e. Minkowski sum)
\begin{align*}
    a_1 \zz + a_2 \zz = \{a + b \mid a \in a_1 \zz, b \in a_2 \zz \}.
\end{align*}
It is well-known that the set $a_1 \zz + a_2 \zz$ can be generated by a single element, which is the gcd of $a_1$ and $a_2$, i.e.
\begin{align*}
    a_1 \zz + a_2 \zz = gcd(a_1, a_2) \zz.
\end{align*}
Furthermore, the set $\mathcal{L}((a_1, a_2)) = a_1 \zz + a_2 \zz$ is closed under addition, subtraction and scalar multiplication, which is why all values for $s,t$ and $r$, as defined above in the Euclidean algorithm, belong to $\mathcal{L}((a_1, a_2))$. In the end, the smallest non-zero element for $r$ obtained by the algorithm generates $\mathcal{L}((a_1, a_2))$ and hence $\mathcal{L}((a_1, a_2)) = r \zz = gcd(a_1,a_2) \zz$.

This interpretation does not only allow for an easy correctness proof of the Euclidean algorithm, it also suggests a generalization of the algorithm into higher dimensions. For this, we consider vectors $A_1, \ldots , A_{n} \in \zz^\dd$ and the set of points in space generated by sums of integral multiples of the given vectors, i.e.
\begin{align*}
    A_1 \zz + \ldots A_{n} \zz.
\end{align*}
This set is called a \emph{lattice} and is generally defined for a given matrix $A$ with column vectors $A_1, \ldots , A_n$ by
\begin{align*}
    \mathcal{L}(A) = \{ \sum_{i=1}^{n} \lambda_i A_i \mid \lambda \in \zz^n \}.
\end{align*}
One of the most basic facts from lattice theory is that every lattice $\mathcal{L}$ has a basis $B$ such that $\mathcal{L}(B) = \mathcal{L}(A)$, where $B$ is a square matrix.\footnote{To be precise, the basis is a square matrix iff the lattice is full dimensional, ie. the span of generating vectors has dimension $d$. In general a basis consists of $\mathrm{rank}(A)$ vectors.}
Note that in this sense the Euclidean algorithm simply computes a basis of the one-dimensional lattice with $gcd(a_1,a_2) \zz = a_1 \zz + a_2 \zz$.

Hence, a multidimensional version of the Euclidean algorithm should compute for a given matrix $A = (A_1, \ldots , A_{n})$ a basis $B \in \zz^{\dd \times \dd}$ such that
\begin{align*}
    \mathcal{L}(B) = \mathcal{L}(A).
\end{align*}
The problem of computing a basis for the lattice $\lattice(A)$ is called \emph{lattice basis computation}. 
In this paper, we show that the classical Euclidean algorithm can be naturally generalized to do just that. Using this approach, we improve upon the running time of existing algorithms for lattice basis computation.

\subsection{Lattice Basis Computation}

Computing a basis of a lattice is one of the most basic algorithmic problems in lattice theory. Often it is required as a subroutine by other algorithms~\cite{DBLP:conf/stoc/Ajtai96,DBLP:conf/eurocal/BuchmannP87,DBLP:conf/stoc/GentryPV08, DBLP:books/daglib/0018102, DBLP:journals/jsc/Pohst87}. There are mainly two methods on how a basis of a lattice can be computed. The most common approaches rely on either a variant of the LLL algorithm~\cite{lll_algorithm} or on computing the Hermite normal form (HNF), where variants of the HNF provide the faster methods. Considering these approaches however, one encounters two major problems. First, the entries of the computed basis can be as large as the determinant and therefore exponential in the dimension. Secondly and more importantly, intermediate numbers on the computation might even be exponential in their bit representation~\cite{DBLP:conf/fct/Frumkin77}. This effect is called intermediate coefficient swell. Due to this problem, it is actually not easy to show that a lattice basis can be computed in polynomial time. Kannan und Bachem~\cite{DBLP:journals/siamcomp/KannanB79} were the first ones to show that the intermediate coefficient swell can be avoided when computing the HNF and hence a lattice basis can actually be computed in polynomial time. The running time of their algorithm was later improved by  Chou and Collins~\cite{DBLP:journals/siamcomp/ChouC82} and Iliopoulos~\cite{DBLP:journals/siamcomp/Iliopoulos89}.

Currently, the most efficient algorithm for computing the Hermite Normal Form (HNF), and consequently for lattice basis computation, is the method developed by Storjohann and Labahn~\cite{DBLP:conf/issac/StorjohannL96}.
Given a full rank matrix $A\in\zz^{\dd\times n}$ the \acs{hnf} can be computed by using $\OTilde(nd^{\omega}\cdot \log\norm{A} )$\footnote{We use the $\OTilde$ notation to omit logarithmic factors, \ie{} $f \in \tilde{O}(g)$ iff $f \in O(g \cdot \log^c(g))$ for some constant $c$. In particular, we omit $\log d$ and $\log \log \norm{A}$ factors.} many bit operations, where $\omega$ denotes the exponent for matrix multiplication and is currently $\omega \leq 2.371552$~\cite{williams2024new_matrixmultiplication}. 
The algorithm by Labahn and Storjohann~\cite{DBLP:conf/issac/StorjohannL96} improves upon a long series of papers~\cite{DBLP:journals/siamcomp/KannanB79, DBLP:journals/siamcomp/ChouC82,DBLP:journals/siamcomp/Iliopoulos89} and, despite being nearly 30 years old, it still offers the best running time available. Only in the special case that $n-d = O(1)$, Li and Storjohann~\cite{DBLP:conf/issac/LiS22} manage to obtain a running time of $\OTilde(d^{\omega}\cdot \log\norm{A} )$ which essentially matches matrix multiplication time. There are other approaches with improved running time for special cases~\cite{pernet2010fast} and/or randomized algorithms~\cite{DBLP:conf/issac/LiuP19, DBLP:journals/talg/BirmpilisLS23}. 

Recent papers considering general lattice basis computation focus on properties of the resulting basis but do not improve the running time. There are several algorithms that preserve orthogonality from the original matrix, \eg{} $\norm{B^*}  \leq \norm{A^*} $, or improve on the $\ell_\infty$ norm of the resulting matrix~\cite{DBLP:conf/stoc/NovocinSV11,DBLP:conf/issac/NeumaierS16}, or both~\cite{DBLP:conf/crypto/HanrotPS11,DBLP:conf/issac/LiN19,DBLP:conf/focs/CaiN97,DBLP:books/daglib/0018102}. 
Except for an algorithm by Lin and Nguyen~\cite{DBLP:conf/issac/LiN19}, all of the above algorithms have a significantly higher time complexity compared to Labahn's and Storjohann's \acs{hnf} algorithm. 
The algorithm by Lin and Nguyen uses existing HNF algorithms and applies a separate coefficient reduction algorithm resulting in a basis with $\ell_\infty$ norm bounded by $d \norm{A} $.

The best size bound on the output basis is $\sqrt{d}\norm{A}$~\cite{DBLP:conf/crypto/HanrotPS11, DBLP:conf/focs/CaiN97, DBLP:books/daglib/0018102, DBLP:conf/issac/LiN19}. This bound was previously achieved via the size reduction known from the LLL algorithm~\cite{lll_algorithm}.
For a more in-depth discussion of the literature on computing lattice bases with bounded size, we refer to Li and Nguyen~\cite{DBLP:conf/issac/LiN19}.

\begin{table}
\renewcommand\thempfootnote{\arabic{mpfootnote}}
\begin{minipage}{\linewidth}
\centering
\begin{tabular}{lll} 
\toprule
\textsc{Algorithm by} & \textsc{Bit Complexity} & \textsc{Basis Size Bound}\smallskip\\  
\midrule 
Kannan and Bachem~\cite{DBLP:journals/siamcomp/KannanB79} & $\tilde{O}(nd^6\log\norm{A})$ & $\det(\lattice(A))$\\

Chou and Collins~\cite{DBLP:journals/siamcomp/ChouC82} & $\tilde{O}(nd^4\log\norm{A})$ & $\det(\lattice(A))$\\

Hafner and McCurley~\cite{DBLP:journals/siamcomp/HafnerM91} & $\tilde{O}(nd^{3} \log\norm{A})$ & $\det(\lattice(A))$\\

Labahn and Storjohann~\cite{DBLP:conf/issac/StorjohannL96} &$\tilde{O}(nd^{\omega}\log\norm{A})$ & $\det(\lattice(A))$\\


Hanrot, Pujol, and Stehlé~\cite{DBLP:conf/crypto/HanrotPS11} & $\tilde{O}(\max\{n,d\}^{8}\log\norm{A})$ & $\sqrt{d}\norm{A}$ \\

CNMG algorithm~\cite{DBLP:conf/focs/CaiN97,DBLP:books/daglib/0018102} & $\tilde{O}(\max\{n,d\}^{4}\log\norm{A})$ & $\sqrt{d}\norm{A}$ \\

Li and Nguyen~\cite{DBLP:conf/issac/LiN19} & $\tilde{O}(\max\{n,d\}^{5}\log\norm{A})$ & $\sqrt{d}\norm{A}$\\

Li and Nguyen~\cite{DBLP:conf/issac/LiN19} & $\tilde{O}(\max\{n,d\}^{\omega+1}\log\norm{A})$ & $d\norm{A}$ \\

Our algorithm & $\tilde{O}(\max\{n-d, d^2\}d^{\omega(2)-1}\log\norm{A})$ & $\sqrt{d}\norm{A}$


\end{tabular}
\caption{Literature overview for lattice basis computation.}
\label{tab:related_work}
\end{minipage}
\end{table}

\subsection{Our Contribution}
In this paper we develop a fundamentally new approach for lattice basis computation given a matrix $A$ with column vectors $A_1, \ldots, A_n \in \zz^d$. Our approach does not rely on any normal form of a matrix or the LLL algorithm. Instead, we show a direct way to generalize the classical Euclidean algorithm to higher dimensions. After a thorough literature investigation and talking to many colleagues in the area, we were surprised to find out that this approach actually seems to be new.

Our approach does not suffer from intermediate coefficient growth and hence gives an easy way to show that a lattice basis can be computed in polynomial time. 

In~\autoref{sec:algorithm_sketch}, we develop an algorithm that chooses an initial basis $B$ from the given vectors and updates the basis according to a remainder operation and then exchanges a vector by this remainder. In every iteration, the determinant of $B$ decreases by a factor of at least $1/2$ and hence the algorithm terminates after at most $\log \det(B)$ many iterations. Our algorithm can be easily described and implemented. The development of the size and the simplicity of the algorithm both are similar to the Euclidean algorithm. 

In~\autoref{sec:fast_euclidean}, we modify our previously developed generalization of the Euclidean algorithm in order to make it more efficient. We exploit the freedom in the pivotization and also combine several iterations at once to optimize its running time. Our main result is an algorithm which requires
\begin{align*}
    LS(\dd,n-\dd,\norm{A})  + \Tilde{O}((n-d)d^2 \cdot \log(\norm{A})),
\end{align*}
many bit operations, where $LS(\dd,n-\dd,\norm{A})$ is the time required to obtain an exact fractional solution matrix $X \in \zz^{\dd \times (n-\dd)}$ to a linear system $B X = C$, with matrix $B \in \zz^{\dd \times \dd}$ and matrix $C \in \zz^{\dd \times (n-\dd)}$ with $\norm{B}, \norm{C} \leq \norm{A}$.
Based on the work of ~\cite{DBLP:conf/issac/BirmpilisLS19}, we argue in~\autoref{sec:lin_sys_solving} that the term $LS(\dd,n-\dd,\norm{A})$, being the bottleneck of the algorithm, can be bounded by 
\begin{align*}
    \tilde{O}(\max\{n-d, d^k\}d^{\omega(k+1)-k}\log\norm{A} )
\end{align*}
for every $k\geq 0$, where $\omega(k)$ is the matrix multiplication exponent of an $n \times n^k$ matrix by an $n^k \times n$ matrix (see~\cite{gall2024faster_rectangularMM, williams2024new_matrixmultiplication}). 
Setting $k=1$, we obtain a running time of $\tilde{O}(\max\{n-d, d\}d^{\omega(2)-1} \log\norm{A} )$, with $\omega(2) \leq 3.250385$~\cite{williams2024new_matrixmultiplication}. This shows that we can improve upon the algorithm by Storjohann and Labhan~\cite{DBLP:conf/issac/StorjohannL96} in any case at least by a factor of roughly $d^{0.12}$ and therefore obtain the first general improvement to this fundamental problem in nearly 30 years. Any future improvement that one can achieve in the computation of exact fractional solutions to linear systems would directly translate to an improvement of our algorithm.

Setting $k=0$ yields a running time of $\tilde{O}((n-d)d^{\omega} \log\norm{A} )$ for our algorithm. Hence, this improves upon the algorithm by Storjohann and Labhan~\cite{DBLP:conf/issac/StorjohannL96} in the case that $n-d \in o(d)$. In the case that $n-d \in O(1)$ we match the running time of the algorithm by Li and Storjohann~\cite{DBLP:conf/issac/LiS22}. Note however, that the algorithm by Li and Storjohann can not be iterated to obtain an algorithm of the same running time for general $n-d > O(1)$. This is due to a coefficient blowup in the output matrix.

Besides the new algorithmic approach, one main tool that we develop is a structural statement on the fractionality of solutions of linear systems. 
Note that when considering bit operations instead of only counting arithmetic operations, one has to pay attention to the bit length of the respective numbers. When operating with precise fractional solutions to linear systems, one can typically only bound the bit length of the numbers by $\Tilde{O}(\dd \log(\norm{A}))$. Hence, we would obtain this term as an additional factor when one is adding or multiplying these numbers. In order to deal with this issue and obtain a better running time, we prove in~\autoref{sec:fractionality} a fundamental structural statement regarding the fractionality (i.e. the size of the denominator) of solutions to linear systems of equations. Essentially, we show that the fractionality of the solutions is only large if there are few integral points contained in a subspace of $\lattice(B)$ and vice versa. This structural statement can then be used in the runtime analysis of our algorithm as it iterates over the integral points in the respective subspace.

In terms of coefficient growth of the output basis matrix, our algorithm is guaranteed to return a solution matrix $S$ such that $\norm{S} \leq d \norm{A}$. In~\autoref{sec:root_n_solution_size} however, we present a postprocessing procedure that further improves upon the size bound of the output basis obtaining $\norm{S} \leq \sqrt{d}/2 \norm{A}$. In contrast to previous techniques relying on the size reduction procedure known from the LLL algorithm
in order to improve upon the size bound of the output matrix, our algorithm relies on a classical result from discrepancy theory which we have not seen applied in this context so far.

Regarding the running time, our postprocessing procedure requires only  $\tilde{O}(d^3 \log \norm{A})$ many bit operations. This running time is only possible due to our structural theorem on the fractionality of solutions as otherwise an additional factor of $d$ would be required. The previously best known algorithm~\cite{DBLP:conf/focs/CaiN97,DBLP:books/daglib/0018102} achieving a bounded output basis $S$ with $\norm{S} \leq \sqrt{d}/2 \norm{A}$ has a running time of $\tilde{O}(\max\{d,n\}^{4}\log\norm{A}_\infty)$. Hence, we improve upon their running time by a factor of at least $\approx d^{0.74}$.


\section{A Multi-Dimensional Generalization of the Euclidean Algorithm}\label{sec:algorithm_sketch}
In this section, unless stated otherwise, we assume that $rank(A) = d$ and therefore the lattice $\lattice(A)$ is full dimensional. However, our algorithms can be applied in a similar way if $rank(A) < d$, which we will discuss later for the refined version in~\autoref{sec:lower_dimensions}. 

\subsection*{Preliminaries}
Consider a (not necessarily full) dimensional lattice $\lattice(B)$ for a given basis $B \in \zz^{\dd \times j}$ with $j \leq \dd$. An important notion that we need is the so called \emph{fundamental parallelepiped}
\begin{align*}
    \Pi(B) = \{ B x \mid x \in [0,1)^j \}
\end{align*}
\begin{figure}
    \centering
    \includegraphics[width=0.2\textwidth]{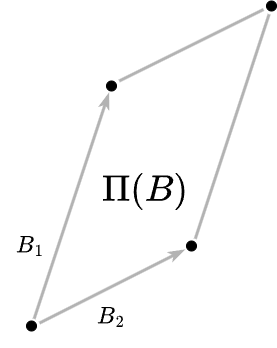}
    \caption{The parallelepiped of $B=(B_1B_2)$.}
    \label{fig:parallelepiped}
\end{figure}
see also~\autoref{fig:parallelepiped}. Let $V$ be the space generated by the columns in $B$. As each point $a \in \mathbb{R}^\dd$ can be written as
\begin{align*}
    a = B \floor{x} + B \{ x \},
\end{align*}
it is easy to see that the space $V$ can be partitioned into parallelepipeds. Here, $\floor{x}$ denotes the vector, where each component $x_i$ is rounded down and $\{ x \} = x - \floor{x}$ is the vector with the respective fractional entries $x_i \in [0,1)$.
In fact, the notion of $\Pi(B)$ allows us to define a multi-dimensional modulo operation by mapping any point $a \in V \cap \zz^d$ to the respective \emph{residue vector} in the parallelepiped $\Pi(B)$, i.e. 
\begin{align*}
    a \pmod {\Pi(B)} := B \{ B^{-1} a \} \in \Pi(B).
\end{align*}
Furthermore, for $a \in \zz^\dd$, we denote with $\nint{a}$ the next integer from $a$, which is $\floor{a + 1/2}$. When we use these notations on a vector $a \in \zz^d$, the operation is performed entry-wise.

Note that in the case that $B \in \zz^{\dd \times \dd}$ the parallelepiped $\Pi(B)$ has the nice property, that its volume as well as the number of contained integer points is exactly $\det(B)$, i.e. 
\begin{align*}
    vol(\Pi(B)) = |\Pi(B) \cap \zz^n| = \det(B).
\end{align*}

In the following we will denote by $A_i$ the $i-th$ column vector of a matrix $A$. In the case that $a \in \zz^\dd$ is a vector we will denote by $a_i$ the $i$-th entry of $a$. To simplify the notation, we denote by $\det (B)$ the absolute value of the determinant of $B$ and by $\norm{B}$ of a matrix $B$ or a vector $B$, we denote the respective infinity norm, meaning the largest absolute entry of $B$.

In our algorithm, we will alter our basis step by step by exchanging column vectors. We denote the exchange of column $i$ of a matrix $B$ with a vector $v$ by $B\setminus B_i \cup v$. The notation $B\cup v$ for a matrix $B$ and a vector $v$ of suitable dimension denotes the matrix, where $v$ is added as another column to matrix $B$. Similarly, the notation $B \cup S$ for a matrix $B$ and a set of vectors $S$ (with suitable dimension) adds the vectors of $S$ as new columns to matrix $B$. While the order of added columns is ambiguous, we will use this operation only in cases where the order of column vectors does not matter.

Our algorithms use three different subroutines from the literature. First, in order to compute a submatrix $B$ of $A$ of maximum rank, we use the following theorem.
\begin{lemma}[\cite{DBLP:conf/issac/LiS22}]\label{lem:maximal_subsystem}
Let $A\in\zz^{\Tilde{d}\times \Tilde{n}}$ have full column rank. There exists an algorithm that finds indices $i_1, \ldots, i_{\Tilde{d}}$ such that $A_{i_1}, \ldots, A_{i_{\Tilde{d}}}$ are linearly independent using $\OTilde(\Tilde{n}\Tilde{d}^{\omega-1}\cdot \log\norm{A})$ bit operations. 
\end{lemma}

As a second subroutine which is required in~\autoref{sec:fast_euclidean} to compute the greatest common divisor of two numbers $a_1,a_2 \leq a$, we use the algorithm by Schönhage~\cite{DBLP:journals/acta/Schonhage71} which requires $\Tilde{O}(\log a)$ many bit operations. Note, that the bit complexity of the classical Euclidean algorithm is actually $\Tilde{O}(\log^2 a)$ as in each iteration of the algorithm, the algorithm operates with numbers having $O(\log a)$ bits.

As a third subroutine, we require an algorithm to solve linear systems of equations of the form $Bx = c$. Due to its equivalence to matrix multiplication this can be done in time $O(d^{\omega})$ time counting only arithmetic operations. However, since we use the more precise analysis of bit operations, we use the algorithm by Birmpilis, Labahn, and Storjohann~\cite{DBLP:conf/issac/BirmpilisLS19} who developed an algorithm analyzing the bit complexity of linear system solving. Their algorithm requires $\Tilde{O}(d^{\omega} \cdot \log \norm{A})$ many bit operations. We modify their algorithm in~\autoref{sec:lin_sys_solving} to obtain an efficient algorithm solving $B X = C$ for some matrix $C$ in order to compute a solution matrix $X$.

\subsection{The Algorithm}
Given two numbers, the classical Euclidean algorithm, essentially consists of two operations. First, a \emph{modulo operation} computes the modulo of the larger number and the smaller number. Second, an \emph{exchange operation} discards the larger number and adds the remainder instead. The algorithm continues with the smaller number and the remainder. 

Given vectors $A=\{A_1, \ldots , A_{\dd+1}\} \subset \zz^\dd$, our generalized algorithm performs a multi-dimensional version of \emph{modulo} and \emph{exchange operations} of columns with the objective to compute a basis $B \in \zz^{\dd \times \dd}$ with $\lattice (B) = \lattice (A)$. First, we choose $\dd$ linearly independent vectors from $A$ which form a non-singular matrix $B$. The lattice  $\lattice(B)$ is a sub-lattice of $\lattice(A)$. Having this sub-basis, we can perform a division with residue in the lattice $\lattice (B)$. Hence, the remaining vector $a \in A \setminus B$ can be represented as
\begin{align*}
    a = B \floor{B^{-1}a} +r,
\end{align*}
where $r$ is the remainder $a \pmod{\Pi (B)}$, see also \autoref{fig:high_dimension_modulo}. 
\begin{figure}[t]
   \centering
      \subfloat[The modulo operation in dimension $2$.
      \label{fig:high_dimension_modulo}]{{\small
    \includegraphics[width=0.40\textwidth]{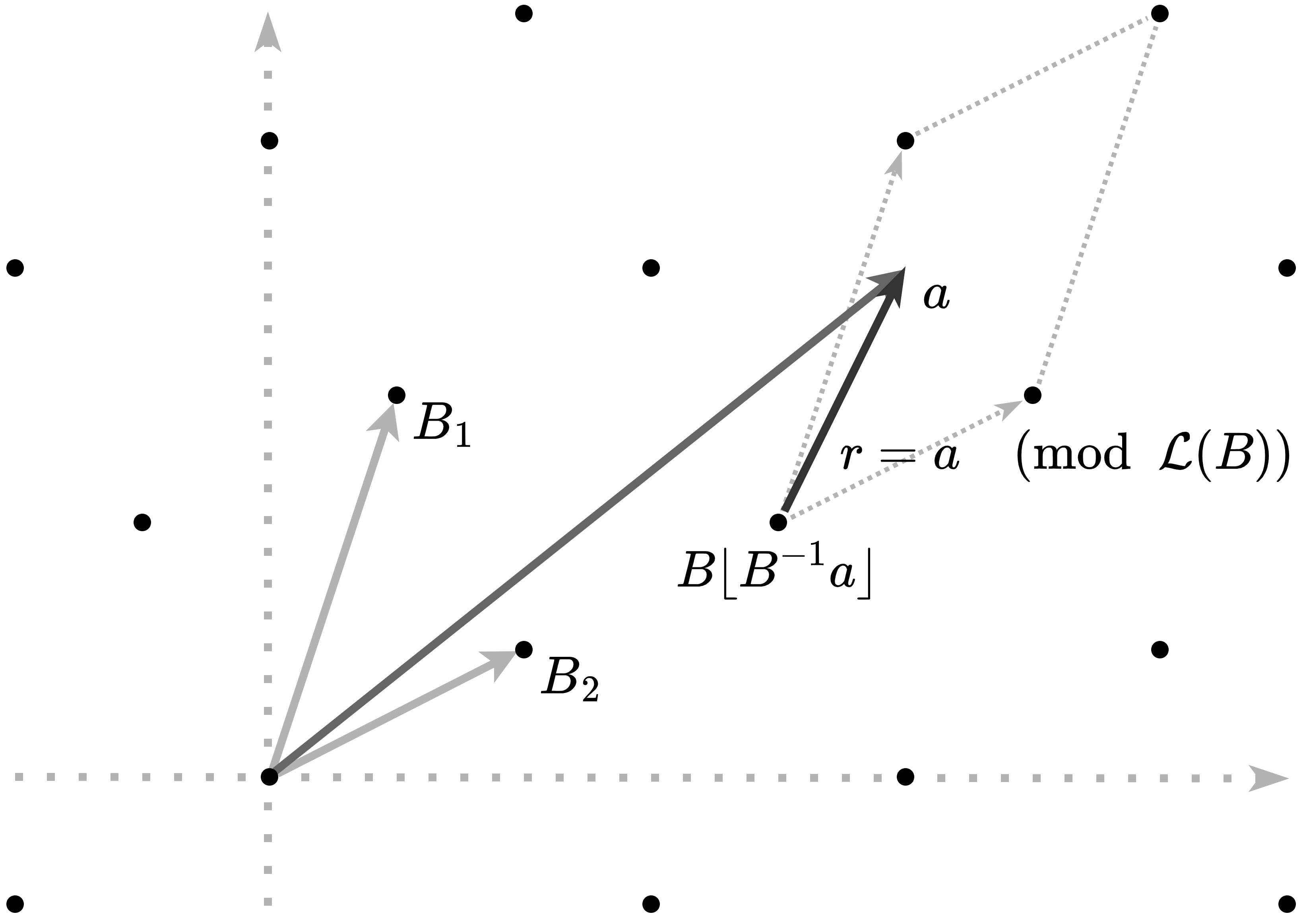}}}\qquad\qquad\qquad
      \subfloat[Exchange of a basis vector and the parallelopipeds for $B_1$ and $B_2$ (solid), $B_2$ and $r$ (dotted), and $B_2$ and $r'$ (dashed).
      \label{fig:basis_exchange}]{{\small
    \includegraphics[width=0.40\textwidth]{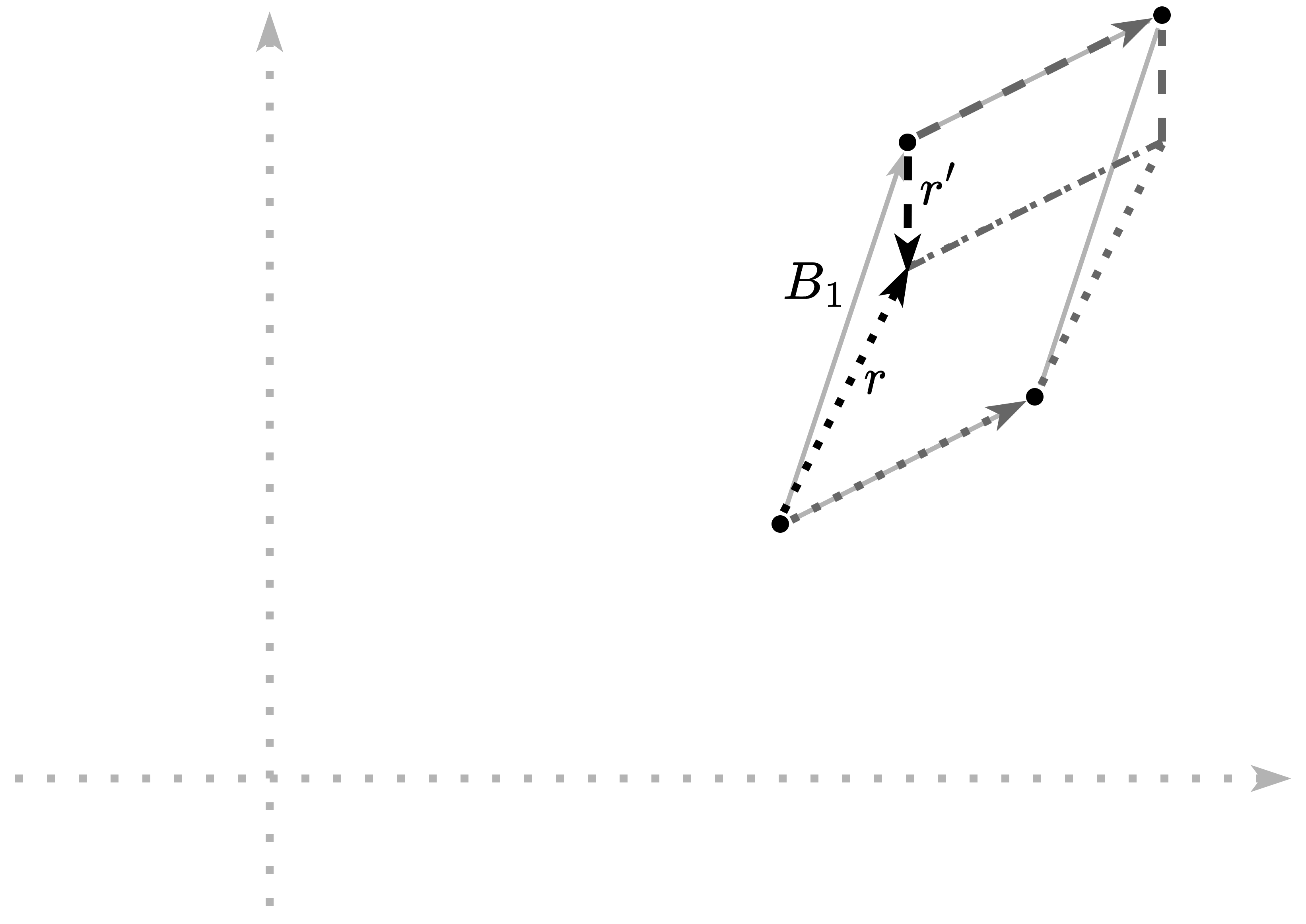}}}\\
    \caption{The modulo operation with respect to a lattice and the exchange operation depending on $\nint{x_1}$.}
\end{figure}
In dimension $d= 1$ this is just the classical division with residue and the corresponding modulo operation, \ie{} $a = b \cdot \floor{a/b} + r$.

Having the residue vector $r$ at hand, the \emph{exchange step} of our generalized version of the Euclidean algorithm exchanges a column vector of $B$ with the residue vector $r$. In dimension $>1$, we have the choice on which column vector to discard from $B$. The choice we make is based on the solution $x \in \qq^\dd$ of the linear system $Bx = a$.
\begin{itemize}
    \item Case 1: $x \in \zz^\dd$. In the case that the solution $x$ is integral, we know that $a \in \lattice (B)$ and hence $\lattice (B \cup a ) = \lattice (B)$. Our algorithm terminates.
    \item Case 2: There is a fractional component $\ell$ of $x$. In this case, our algorithm exchanges $B_i$ with $r$, \ie{} $B' = B \setminus  B_\ell  \cup  r $.
\end{itemize}
The algorithm iterates this procedure with basis $B'$ and vector $a = B_\ell$ until Case 1 is achieved.\medskip

\begin{addmargin}[3.5em]{3.5em}
\begin{tcolorbox}[
sharp corners=all,
colback=white,
colframe=black,
size=tight,
boxrule=0.2mm,
left=3mm,right=3mm,top=3mm,bottom=3mm
]
{\begin{multicols}{2}

\textbf{Euclidean Algorithm}\bigskip

\textsc{Modulo Operation}\\
$t = s\lfloor s^{-1} t\rfloor + r$ \medskip

\textsc{Exchange Operation}\\
$t = s, \ s=r$ \medskip

\textsc{Stop Condition}\\
$s^{-1} t\ $ is integral

\columnbreak

\textbf{Generalized Euclidean Algorithm}\bigskip

\textsc{Modulo Operation}\\
$a = B\floor{B^{-1}a} + r$ \medskip

\textsc{Exchange Operation}\\
$a = B_\ell, \ B_\ell := r$ \medskip

\textsc{Stop Condition}\\
$B^{-1}a\ $ is integral

\end{multicols}}
\end{tcolorbox}
\end{addmargin}

Two questions arise: Why is this algorithm correct and why does it terminate?

\textbf{Termination:}\\
The progress in step 2 can be measured in terms of the determinant. For $x$ with $Bx=a$ the exchange step in case 2 swaps $B_i$ with $r = B \{ x \}$ and $\{x_i\} \neq 0$ to obtain the new basis $B'$. By Cramer's rule we have that $\{x_\ell\} = \frac{\det B'}{\det B}$ and hence the determinant decreases by a factor of $\{x_\ell\} < 1$. The algorithm eventually terminates since $\det(\lattice(A))\geq 1$ and all involved determinants are integral since the corresponding matrices are integral. A trivial upper bound for the number of iterations is  the determinant of the initial basis. 

\textbf{Correctness:}\\
Correctness of the algorithm follows by the observation that $\lattice(B \cup a) = \lattice(B \cup r)$. To see this, it is sufficient to prove $a \in \lattice(B \cup r)$ and $r \in \lattice(B \cup a)$. By the definition of $r$ we get that $a = Bx = B\floor{x} + B\{x\} = B\floor{x} + r$. Hence, $a$ and $r$ are integral combinations of vectors from $B\cup r$ and $B\cup a$, respectively, and hence $\lattice(B\cup a) = \lattice(B \cup r)$.


The multiplicative improvement of the determinant in step 2 can be very close to $1$, \ie{} $\frac{\det(B)-1}{\det(B)}$. In the classical Euclidean algorithm a step considers the remainder $r$ for $a = b\floor{a/b} + r$. The variant described in \autoref{sec:intro} considers an $r'$ for $a = b\nint{a/b} + r'$. Taking the next integer instead of rounding down ensures that in every step the remainder in absolute value is at most half of the size of $b$. Our generalized Euclidean algorithm uses a modified modulo operation that does just that in a higher dimension. In our case, this modification ensures that the absolute value of the determinant decreases by a multiplicative factor of at most $1/2$ in every step as we explain below. The number of steps is thus bounded by $\log\det(B)$. The generalization to higher dimensions chooses $i$ such that $x_i$ is fractional and rounds it to the next integer $\nint{x_i}$ while the other entries of $x$ are again rounded $\floor{x_j}$ for $j\neq i$. Formally, this modulo variant is defined as 
\begin{align*}
    a \ \ (\bmod'\ \Pi(B)) := r' := a - (\sum_{j\neq i}B_j\floor{x_j} + B_i\nint{x_i})
\end{align*}
for $Bx=a$ and some $i$ such that $\{x_i\}\neq 0$. By Cramer's rule we get that the determinant decreases by a multiplicative value of at least $1/2$ in every iteration since $\frac{1}{2} \leq \abs{x_i - \nint{x_i}} = \abs{\frac{\det B'}{\det B}}$. 

In \autoref{fig:basis_exchange} the resulting basis for exchanging $B_1$ with $r = a \ \ (\bmod\  \Pi(B))$ and with $r'= a \ \ (\bmod'\  \Pi(B))$ shows that in both cases the volume of the parallelepiped decreases, which is equal to the determinant of the lattice. In \autoref{fig:application}, an example of our algorithm is shown.
\begin{figure}[t]
   \centering
      \subfloat[Application of our algorithm, $r'$ is the first remainder.
      \label{fig:alg_0}]{{\small
    \includegraphics[width=0.30\textwidth]{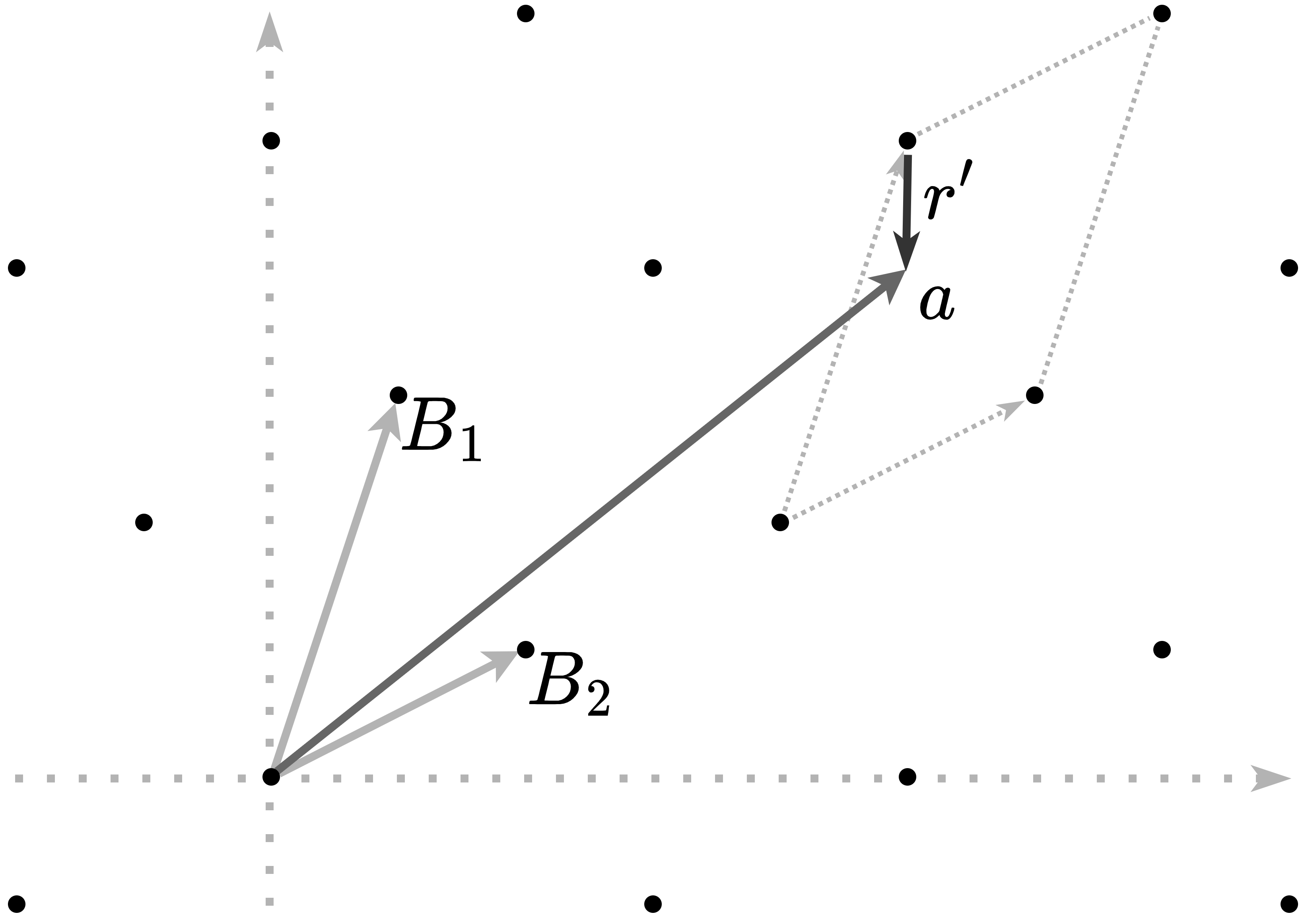}}} 
    \qquad
      \subfloat[Vectors $r'$ and $B_1$ were exchanged and $r''$ denotes the second remainder.
      \label{fig:alg_1}]{{\small
    \includegraphics[width=0.30\textwidth]{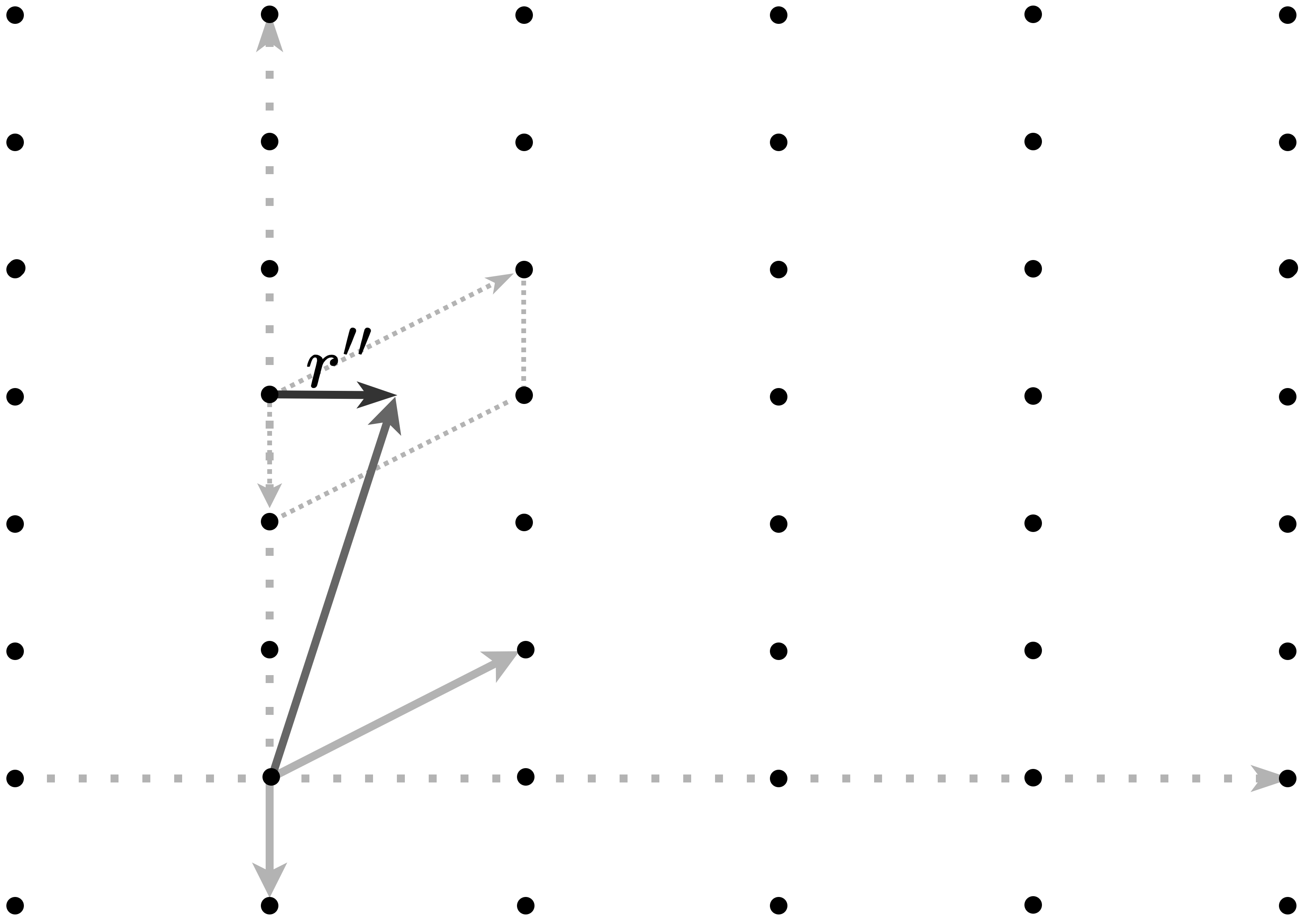}}}\qquad
      \subfloat[Vectors $r'$ and $B_2$ were exchanged. $B_2$ is in the lattice and the algorithm terminates.
      \label{fig:alg_2}]{{\small
    \includegraphics[width=0.30\textwidth]{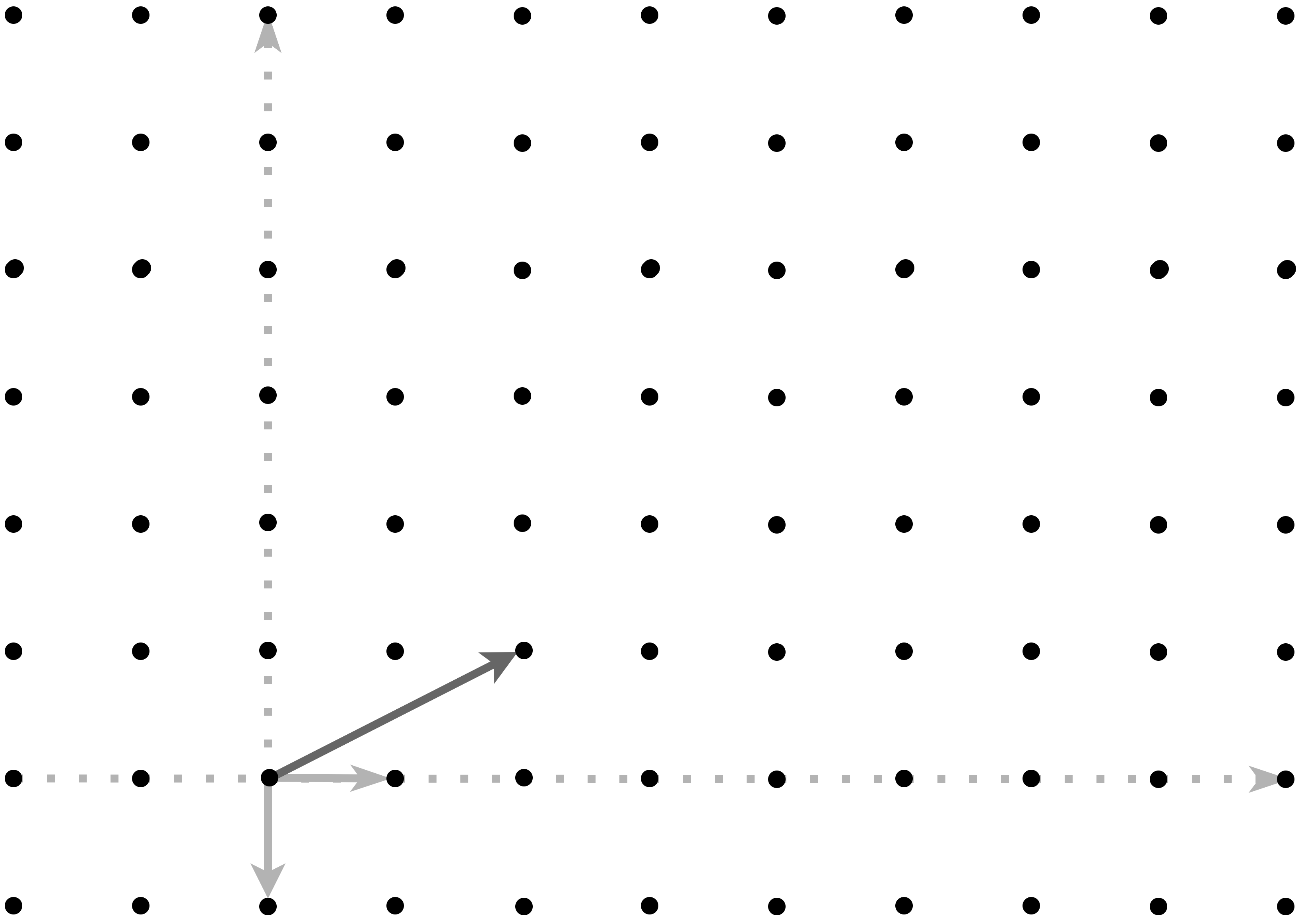}}}
   \caption{An application of the algorithm.}
    \label{fig:application}
\end{figure}


\subsection{Formal Description of the Algorithm}
In the following, we state the previously described algorithm formally.
\begin{algorithm}
  \caption{Generalized Euclidean Algorithm (Basic Algorithm)
    \label{alg:easyversion}}
  \begin{algorithmic}[1]
    \Statex
    \textsc{Input: } A matrix $A=\left(A_1,\ldots, A_n \right) \in \zz^{\dd\times n}$ \\
    \textbf{find} independent vectors $B := \left(B_1, \ldots, B_\dd\right)$ with $B_i\in \{A_1,\ldots, A_n\}$
    \Let{$C$}{$\{A_1,\ldots, A_n\}\setminus \{B_1,\ldots, B_\dd\}$}
    \While{$C\neq \emptyset$}
      \State{\textbf{choose any} $c\in C$}
      \State{\textbf{solve} $Bx=c$ }
      \If{$x$ is integral}
      \Let{$C$}{$C\setminus\{c\}$}
      \Else
      \State{\textbf{choose} index $\ell \leq d$ with $\{x_\ell\}\neq 0$}
      \Let{$C$}{$C\setminus\{c\}\cup \{B_\ell\}$}
      \Let{$B_\ell$}{$c - (B_\ell \nint{x_\ell} + \sum_{j\neq \ell}B_j\floor{x_j})$}
      
      \EndIf
    \EndWhile
    \State \Return{$B$}
  \end{algorithmic}
\end{algorithm}

\begin{theorem}\label{theorem:basic_version_correctness}
\autoref{alg:easyversion} computes a basis for the lattice $\mathcal{L}(A)$.
\end{theorem} 
\begin{proof}
Let us consider the following invariant.\medskip

\textit{Claim.} In every iteration $\mathcal{L}(A) = \mathcal{L}(B\cup C)$.

By the definition of $B$ and $C$ the claim holds in line 2. We need to prove that removing $c$ from $C$ in line $7$ and altering $B$ and $C$ in lines 9-11 do not change the generated lattice. In line 7 we found $c$ is an integral combination of vectors in $B$. Thus, every lattice point can be represented without the use of $c$ and $c$ can be removed without altering the generated lattice. In lines 9 and 10 a vector $c$ is removed from $B\cup C$ and instead a vector $c' = c - (\sum_{j\neq i}B_j\floor{x_j}  + B_i \nint{x_i})$ is added. By the definition of $c'$, the removed vector $c$ is an integral combination of vectors $c', B_1, \ldots, B_n$ and $c'$ is an integral combination of vectors $c, B_1, \ldots, B_n$. Using the same argument as above, this does not change the generated lattice. \medskip

The algorithm terminates when $C=\emptyset$. In this case $B$ is a basis of $\mathcal{L}(A)$, since by the invariant we have that $\mathcal{L}(B) = \mathcal{L}(B \cup C) = \mathcal{L}(A)$.

\end{proof}

The following observation holds as in each iteration the coefficient for the newly added row $B_\ell$ in l.11 equals $x_\ell - \lceil x_\ell \rfloor$ and hence in absolute value is bounded by $1/2$. By Cramer's rule, this implies that $\frac{\det B'}{\det B} \leq \frac{1}{2}$, where $B'$ is the new matrix being defined by interchanging the $\ell$-th column in l.11. Hence, $\det (B') \leq \frac{1}{2} \det (B)$ which implies the following observation.
\begin{observation}
\autoref{alg:easyversion} terminates after at most $\log\det(B)$ exchange steps, where $B$ is the matrix defined in l.1 of the algorithm.
\end{observation}
Note for the running time of the algorithm that in each iteration the algorithm solves a linear system $Bx = c$ as its main operation. Therefore, we obtain a running time of
\begin{align*}
        LS \cdot (\log (\det B) + (n-\dd)),
\end{align*}
where $LS$ is the time required to solve the linear system $Bx = c$. Note however, that the coefficient in $B$ might grow over each iteration. Since the new vector is contained in $\Pi(B)$ this growth can be bounded by $d \norm{B}$. Hence, the resulting matrix at the end of~\autoref{alg:easyversion} has entries of size at most $d^{\log (\det B)} \cdot \norm{B}$.

Using the Hadamard bound $\det(B) \leq d \norm{B}^d$ and using that $LS$ can be bounded by $\Tilde{O}(d^\omega \cdot \log \norm{C})$ for some matrix $C \in \zz^{\dd \times \dd}$ \cite{DBLP:conf/issac/BirmpilisLS19}, we obtain that~\autoref{alg:easyversion} requires at most 
\begin{align*}
    \Tilde{O}(d^{\omega+1} \log (\det B)\cdot(\log (\det B) + (n-d)) ) \leq \Tilde{O}(n d^{\omega+2} \cdot \log^3 (\norm{A}))
\end{align*}
many bit operations as a first naive upper bound for the running time of~\autoref{alg:easyversion}. We improve upon the running time as well as the size of the entries of the output matrix in the following section.

Note however, that~\autoref{alg:easyversion} does not necessarily require an exact fractional solution in l.5. It is sufficient to decide if an index of the fractional solution $x$ is in indeed fractional and the first fractional bit in order to decide on how to round. This makes an implementation of the algorithm very easy as an out of the box linear system solver can be used.

In an earlier version of the paper~\cite{klein2023old_paper}
we described methods and data structures on how~\autoref{alg:easyversion} can be modified to be more efficient.


\section{Speeding Up the Generalized Euclidean Algorithm}\label{sec:fast_euclidean}
In the generalized setting of the Euclidean algorithm as defined in the previous section, we have several degrees of freedom. First, in each iteration, we may choose an arbitrary vector $c \in C$. Secondly, we may choose any arbitrary index $i$ with $\{ x_i \}$ being fractional to pivot and iterate the algorithm. We can use this freedom to optimize two things:
\begin{quote}
  \begin{itemize}
    \item the running time of the algorithm and
    \item the quality of the solution, i.e. the size of the matrix entries in the output solution.
  \end{itemize}
\end{quote}

In the following, we study one specific rule of pivotization which allows us to update the solution space very efficiently. Moreover, it allows us to apply the Euclidean algorithm, or respectively Schönhage's algorithm~\cite{DBLP:journals/acta/Schonhage71}, as a subroutine. The main idea of the following algorithm is to pivot always the same index until it is integral for all $c \in C$. Proceeding this way allows us to update the solution very efficiently by applying the one-dimensional classical Euclidean algorithm. In the next iteration, the algorithm continues with one of the remaining fractional indices.

Consider a basis $B \in \zz^{\dd \times \dd}$, and a set of vectors $C \in \zz^{\dd \times (n-\dd)}$ which need to be merged into the basis. Let $X_c \in \qq^\dd$ with $B X_c = c$ be the solution vector for each $c \in C$. For a fixed index $\ell \leq \dd$, every entry equals $(X_c)_\ell = \frac{t_c}{\notd}$ for some numbers $t_c \in \zz$ which we call the \emph{translate} of vector $c$ and some $\notd \in \zz$ with $\notd \leq \det(B)$. Intuitively, the translate $t_c$ for a vector denotes the distance between vector $c \in C$ and the subspace $V_\ell$ which is generated by the columns $B \setminus B_\ell$. Clearly, if $t_c = 0$ then the vector $c$ is contained in the subspace $V_\ell$. How close can a linear combination of vectors in $C$ and $B_\ell$ come to the subspace $V_\ell$? The closest vector lies on the translate $g$ being the greatest common divisor of all $t_c$ and $\notd$, with $\notd$ being the translate of $B_\ell$. 

Let $S_\ell \in \lattice(A)$ be a vector on translate $g$. Three easy observations follow:
\begin{itemize}
    \item No vector in the lattice $\lattice(A)$ can be contained in the space between $V_\ell$ and its translate $S_\ell + V_\ell$.
    \item Having vector $S_\ell$ at hand, we can compute a vector on every possible translate of $V_\ell$ simply by taking multiples of $S_\ell$.
    \item Having a basis $B'$ for the sublattice $\lattice(A) \cap V_\ell$, the basis $B' \cup S_\ell$ generates the entire lattice $\lattice(A)$.
\end{itemize}

In order to compute a basis $B'$ for the sublattice $\lattice(A) \cap V_{\ell}$ the algorithm, continues iteratively with vectors being contained in the subspace $V_\ell$. To obtain those vectors, the algorithm subtracts from each vector $c \in C$ a multiplicity of the previous vector $t_{c'}$ such that the resulting vectors lie on the translate $0$, i.e. the subspace $V_{\ell}$.

In the example given in~\autoref{fig:fast_euclidean}, we consider the subspace generated by $B_1$. For the translates we obtain that $t = 9$ as $B_2$ is on the $9$-th translate. The vectors $C_1$ and $C_2$ are contained in the $6$-th and $4$-th translate respectively. Hence, the vector $Y_1 = B_2 - 2 C_1 + C_2$ is contained on the first translate.
\begin{figure}
    \hfill
    \begin{tikzpicture}[scale=0.5]
 
 \def\X{6}
 \def\Y{8}
 
 
 
 \foreach \x in {0,...,\X}{
   \foreach \y in {0,...,\Y}
     \node[draw,inner sep=0,circle, fill] at (\x,\y) {};
 }
 

 \draw[-] (0,0) -- (6,3) node[below] {$B_1$};
 \draw[-] (0,0) -- (1,5) node[left] {$B_2$};
 
 \draw[-] (6,3) -- ++(1,5);
 \draw[-] (1,5) -- ++(6,3);

 \draw[dotted] (1/9,5/9) -- ++(6,3);
 \draw[dotted] (2/9,10/9) -- ++(6,3);
 \draw[dotted] (3/9,15/9) -- ++(6,3);
 \draw[dotted] (4/9,20/9) -- ++(6,3);
 \draw[dotted] (5/9,25/9) -- ++(6,3);
 \draw[dotted] (6/9,30/9) -- ++(6,3);
 \draw[dotted] (7/9,35/9) -- ++(6,3);
 \draw[dotted] (8/9,40/9) -- ++(6,3);
 

 \fill[red] (2,4) circle (3pt) node[below] {$C_1$};
 \fill[red] (4,4) circle (3pt) node[below] {$C_2$};

 \fill[red] (1,1) circle (3pt) node[above] {$S_1$};

\end{tikzpicture}
    \hfill
    \begin{tikzpicture}[scale=0.5]
 
 \def\X{6}
 \def\Y{8}
 
 
 
 \foreach \x in {0,...,\X}{
   \foreach \y in {0,...,\Y}
     \node[draw,inner sep=0,circle, fill] at (\x,\y) {};
 }
 

 \draw[-] (0,0) -- (6,3) node[below] {$B_1$};
 


 \fill[red] (2,1) circle (3pt) node[above] {$C_{1}'$};
 \fill[red] (0,0) circle (3pt) node[above] {$C_{2}'$};


\end{tikzpicture}
    \hfill\quad
    \caption{Application of Algorithm~\ref{alg:fasteuclidean}}
    \label{fig:fast_euclidean}
\end{figure}
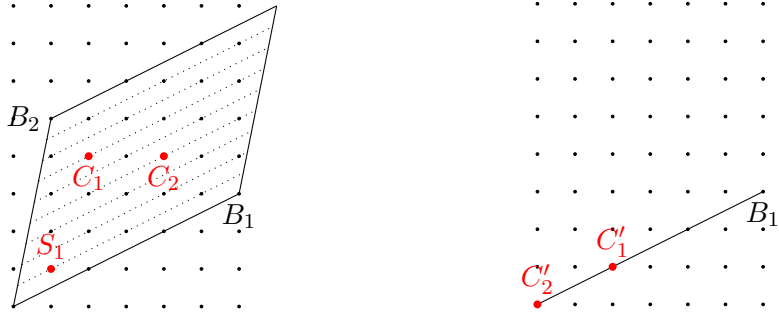
The algorithm then continues iteratively with the vectors $C_{1}' \equiv C_1 - 4 S_1 \pmod {\Pi(B)}$ and the vector $C_{2}' \equiv C_2 - 6 S_1 \pmod {\Pi(B)}$ being contained in the subspace $V_2$.

In the following, we state this algorithmic approach formally. For the correctness of the algorithm, 
one can ignore l.5 for now. Instead, one may assume that the algorithm chooses an arbitrary index $\ell$. The running time of the algorithm is analyzed later in~\autoref{sec:runtime_fast_euclid}. 

\begin{algorithm}
  \caption{Fast Generalized Euclidean Algorithm
    \label{alg:fasteuclidean}}
  \begin{algorithmic}[1]
    \Statex{
    \textsc{Input: } A matrix $A=\left(A_1,\ldots, A_n \right) \in \zz^{\dd\times n}$ }
    \State{\textbf{find} independent vectors $B := \left(B_1, \ldots, B_\dd\right)$ with $B_i\in \{A_1,\ldots, A_n\}$}
    \State{\textbf{let} $C$ be a matrix with columns $\{A_1,\ldots, A_n\}\setminus \{B_1,\ldots, B_\dd\}$}
    \State{\textbf{solve} $B X = C$}
    \For{$i=1$ to $\dd$}
        \State{\textbf{choose} index $\ell$ with maximum fractionality in matrix $X$}
       \State{\textbf{determine} the translate $\notd$ of $B_\ell$ and the translate $t_i$ of vectors $C = (C_1, \ldots , C_{n-d})$}
       \Statex{\hspace{2.65cm}with respect to the subspace $B \setminus B_\ell$}
        \For{$j=1$ to $n-d$} 
        \State{\textbf{compute} $g_j := gcd(t, t_1, \ldots , t_j)$ and factors $\alpha_j, \beta_j$ such that $g_j = \alpha_j g_{j-1} + \beta_j t_j$} 
       \State{\textbf{set} $Z_j = \alpha_j Z_{j-1} + \beta_j X_{C_j}$ with $Z_0 = e_{\ell}$}
       \State{\textbf{set} $X_{C_j} = g_{j-1}/g_j \cdot X_{C_{j}} - t_{j}/g_j \cdot Z_{j-1}$}
       \EndFor
       \State{\textbf{set} $Y_i = Z_{n-d}$}
    \EndFor
    \State \Return{Basis $S = B Y$, with matrix $Y$ containing $Y_i$ as the $i$-th column}
  \end{algorithmic}
\end{algorithm}

We want to give some intuition on the correctness of~\autoref{alg:fasteuclidean}. First, note that instead of solving $Bx =c$ in l.5 of \autoref{alg:easyversion} from scratch after each pivoting step, we can instead update the previous solution. The following rule updates the solution vector if $c$ is being exchanged with column vector $B_\ell$.
\begin{align}\label{rule:update}
    x_i = \begin{cases} \frac{1}{x_i} & i=\ell \\
    - \frac{x_i}{x_\ell} & i\not=\ell \end{cases}
\end{align}
This update rule can be easily seen as 
\begin{align*}
    \sum_{i=1}^d B_i x_i = c \Leftrightarrow -c + \sum_{i=1, i\not=\ell}^d B_i x_i  = - B_\ell x_\ell \Leftrightarrow
    \frac{1}{x_\ell}c - \sum_{i=1, i\not=\ell}^d B_i \frac{x_i}{x_\ell}  = B_\ell
\end{align*}
Using the update rule~(\ref{rule:update}), we can now track how the solution $x$ changes over each iteration if one always pivots the same index $\ell$ and hence exchanges the residue of vector $c$ with column vector $B_\ell$ until $x_\ell$ is integral. Let $x_\ell = \frac{a}{b}$ then $x_\ell$ is being changed in the following way:
\begin{itemize}
    \item Modulo Computation: $B x' = \{c \}$ implies $x'_{\ell} = x_\ell - \lfloor x_\ell \rfloor = \frac{R_1}{b}$ with $R_1 = a \mod b$.\footnote{Here we round down instead of to the closest integer for easier notation. Note that both versions are equivalent.}
    \item Exchange Step: $B'x'' = B_\ell$ implies $x''_{\ell} = \frac{1}{x'_{\ell}}= \frac{b}{R_1}$, where $B' = B\setminus \{B_\ell\} \cup \{\{c\}\}$.
\end{itemize}
Iterating the above process, the numerator and denominator evolve as follows
\begin{align*}
        a &= f_1 b + R_1\\
        b &= f_2 R_1 + R_2\\
        & \ \, \vdots \\
        R_{r-1} &= f_r R_{r} + R_{r+1}\\
        R_{r} &= f_{r+1}R_{r+1} + 0 
    \end{align*}
for residues $R_2, \ldots, R_{r+1} \in \zz$ and factors $f_1, \ldots , f_{r+1} \in \zz$ with $g = R_{r+1}$ being the gcd of $a$ and $b$. This is equivalent to performing an the classical Euclidean algorithm applied to $a,b \in \zz$.
Consequently, Algorithm~\ref{alg:easyversion} performs equivalently to the classical Euclidean algorithm with respect to $x_\ell = \frac{a}{b}$ in the case that one is always pivoting the same index $\ell$ until $x_\ell$ is integral.

Now, the same execution of the Euclidean algorithm can be applied to the vector $x \in \qq^d$ with $Bx = c$ and $e_\ell$ (being the $\ell$-th unity vector) instead of the numbers $a,b \in \zz$ -- obtaining identical factors $f_1, \ldots f_{r+1} \in \zz$. However, instead of the respective residues $R_{1}, \ldots R_{r+1} \in \zz$ we obtain fractional residue vectors $x[1], \ldots , x[r+2]$ with $Bx[i]$ being the vector lying on translate $R[i]$. We obtain the following equivalence chain
\begin{align}\label{eq:chainvectors}
        x &= f_1 e_\ell + x[1] \nonumber\\
        e_\ell &= f_2 x[1] + x[2]\\
        & \ \, \vdots \nonumber \\
        x[r-1] &= f_r x[r] + x[r+1]\nonumber \\
        x[r] &= f_{r+1}x[r+1] + x[r+2] \nonumber
\end{align}
with $x_\ell[i] = \frac{R_i}{R_{i-1}}$. Vector $x[i]$ is the solution vector that we would obtain after by applying Algorithm~\ref{alg:easyversion} after $i$ iterations.
Hence the solution vector $Bx[r+1]$ lying on the gcd translate $g$ is the vector that we obtain after pivoting by the same index $\ell$ until it is integral. This is the vector we are actually interested in and here is the way on how we can compute it efficiently: Instead of performing the equivalence chain of vectors in~(\ref{eq:chainvectors}), we can apply the Extended Euclidean algorithm to obtain factors $\alpha, \beta \in \zz$ such that the gcd $g = R_{r+1} = \alpha a + \beta b$. We can hence obtain the vector $x[r+1]$ by 
\begin{align*}
    x[r+1] = \alpha x + \beta e_\ell.
\end{align*}
Consequently, the vector $x[r+2] = b/g\ x - a/g\ e_\ell \in \qq^d$ is the vector lying on translate $0$.

Iterating this process over all vectors $c \in C$ and and all fractional indices $\ell$ yields the following algorithm. As we have argued, its correctness is already implied by the correctness of Algorithm~\ref{alg:easyversion} since we only simulate several iterations at once.

While the correctness of the algorithm can be seen by the argumentation above, in the following we give an additional proof which is more formal and does not rely on the correctness of Algorithm~\ref{alg:easyversion}.
\begin{theorem}\label{thm:fast_euclid_correct}
    ~\autoref{alg:fasteuclidean} returns a basis $S \in \zz^{\dd \times \dd}$ with \begin{align*}
        \lattice (A) = \lattice (S),
    \end{align*}
    where $S = BY$ is the solution matrix returned at the end of~\autoref{alg:fasteuclidean}.
\end{theorem}
\begin{proof}
    Without loss of generality, assume that the algorithm pivots index $i$ in the $i$-th iteration.
    Let $B^{(i)}\in \zz^{d \times (\dd-i+1)}$ be the basis matrix containing column vectors $B_i, \ldots , B_\dd$, i.e. 
    \begin{align} \label{def:Bi}
        B^{(i)} = \begin{pmatrix}
            B_i, \ldots, B_\dd
        \end{pmatrix}.
    \end{align}
    Furthermore, let $C[i]$, be the image of the solution vectors $X_c$ for each $c \in C$ as defined at the beginning of the $i$-th iteration of the algorithm, i.e.
    \begin{align} \label{def:Ci}
        C[i] = \{ B (X_c[i]) \mid c \in C \},
    \end{align}
    where $X_c[i] \in \qq^d$ is the vector defined at l.3 and updated in l.10 in the $i$-th iteration of the algorithm. By definition it holds that $C[0] = C$. Similarly, let $G[i]$ be the image of $Z_j$ as defined in the $i$-th iteration in l.9, \ie{} $G[i] = \{BZ \mid Z = (Z_1, \ldots, Z_{n-d})\}$ and, for simplified notation let $G_0[i] = B_i$.

    As a first observation, note that the vector $G_j[i]$ lies on translate $g_{j}$ (as defined in the $i$-th iteration). This implies that the vectors in $C[i]$ belong to the subspace generated by vectors in $B^{(i+1)}$. This is because $C_j[i]$ is on translate $t_j$ and the vector $G_{j-1}[i]$ is on translate $g_{j-1}$. Hence the resulting vector $C_j[i+1] = g_{j-1}/g_j \cdot X_{C_{j}[i]}[i] - t_{j}/g_j \cdot Z_{j-1}$ must be on translate $0$.

    \textbf{Observation 1:} Each Vector in $C[i]$ belongs to the subspace generated by vectors in $B^{(i+1)}$.

    The following observation follows directly using that $Y_i$ is defined by the last gcd translate $G_{n-d}[i]$:
    
    \textbf{Observation 2:} Each vector $S_i = B Y_i$ belongs to the subspace generated by vectors in $B^{(i)}$ and is on translate $g_{n-d}$ being the gcd of all translates $t, t_1, \ldots , t_{n-d}$.

    \textbf{Observation 3:} $\lattice(C_{j}[i-1] \cup G_{j-1}[i-1]) = \lattice(C_j[i] \cup G_j[i])$.\\ 
    The inclusion $\lattice(C_{j}[i-1] \cup G_{j-1}[i-1]) \supseteq \lattice(C_j[i] \cup G_j[i])$ follows directly from the definitions in l.9 and l.10 as $G_j[i]$ is defined by an integral combination of $G_{j-1}[i]$ and $C_j[i-1]$. 
    For the other direction the argument follows similarly to \autoref{theorem:basic_version_correctness} as $C_j[i] \cup G_j[i]$ result from two consecutive steps in the Extended Euclidean algorithm on the translates as $G_j[i]$ is on the gcd of translates of $C_{j}[i-1]$ and $G_{j-1}[i]$ and $C_j[i]$ results from the following step for the zero-th translate. The factors $g_{j-1}/g_j$ and $t_{j}/g_j$, which we use to compute $C_j[i]$, are well-known to represent $0$ in the last step and of course the gcd is the second to last step. Now, since we have two consecutive steps in the Euclidean algorithm, we can restore all previous vectors (as integral combinations) using the remainder calculation as
    \begin{align*}
        C_j[i-1] &= f_1 G_{j-1}[i] + R_1\\
        G_{j-1}[i] &= f_2 R_1 + R_2\\
        & \ \, \vdots \\
        R_{r-1} &= f_r R_{r} + G_j[i]\\
        R_{r} &= f_{r+1}G_j[i] + C_j[i] 
    \end{align*}
    for some residue vectors $R_1, \ldots , R_{r} \in \zz^d$ and integral factors $f_1, \ldots , f_{r+1} \in \zz$ which are the quotients in the Extended Euclidean algorithm for $gcd(g_{j-1}, t_j)$. As every previous vector is an integral combination of the following two, we also get $C_j[i-1],G_{j-1} \in \lattice(C_j[i] \cup G_j[i])$ and thus the third observation holds. 

    Using the third observation iteratively for the inner for-loop in l.7-10, we get that $\lattice(B_i \cup C[i-1]) = \lattice(BY_{i} \cup C[i])$. Moreover, this implies $\lattice(B \cup C[0]) = \lattice(S \cup C[d])$ when we iterate the above over all $d$ iterations of the outer for-loop in l.4-11. From the first observation we have that $C[d]$ are all zero as the subspace $B^{(d+1)}$ is just zero. Hence, $\lattice(S) = \lattice(B \cup C) = \lattice(A)$.

\end{proof}

In the following we show that \autoref{alg:fasteuclidean} can be applied in a way such that all fractional entries in the solution matrices $X$ and $Y$ and vectors $Z_j$ (l.3 and l.9-11 of the algorithm) are being computed modulo $1$ (i.e. we cut off the integral part of each fractional entry). Non-zero integral numbers are being mapped to $1$. This modification has two purposes. The first is to bound intermediate numbers, allowing us to assume that each entry of the respective matrices consists of a fractional number $\frac{a}{b}$ with $a,b \in \zz$ and $|a| \leq |b|$. The second purpose is to bound the size of the resulting basis.

\begin{lemma}\label{lem:modulo_entries}
    Entries of the matrices $X$ and $Y$ and vectors $Z_j$ in \autoref{alg:fasteuclidean} can be computed modulo $1$ (as described above) without compromising the correctness of the algorithm.
\end{lemma}
\begin{proof}
    Regarding $X$ in l.3, it is easy to see that $\lattice(B \cup C) = \lattice(B \cup BX')$, where $X'$ is the matrix, where each entry $X'_{ij}$ is being set to $\{ X_{ij} \}$. This is due to the original vectors $C$ being representable by integral combinations of vectors of $B$ and $BX$ as $C = BX' + B(X - X')$.

    Without loss of generality, assume that the algorithm pivots index $i$ in the $i$-th iteration and let $B^{(i)}, X[i]$ and $C[i]$ be defined as in the proof of Theorem~\ref{thm:fast_euclid_correct}.

    Now it remains to show the claim for updates of $X$, $Z_j$, and $Y$ in l.9-11. Consider iteration $i$ of the loop and note that $X[i]$ has zero-entries in all rows that correspond to previous iterations. 

    We want to show that $\lattice(B^{(i+1)} \cup S_i \cup C[i+1]) = \lattice(\hat{B}^{(i+1)} \cup \hat{S}_i \cup \hat{C}[i+1])$ where the latter applies to the altered modulo in l.9-10 and the former does not. Since the vectors of $B^{(i+1)}$ are considered and added to the basis in  subsequent iterations, including them in the above equation is still a sufficient claim. The altered modulo applied to indices $> i$ do not change the lattice as they can be easily reconstructed using the vectors of $B^{(i+1)}$. For indices $<i$ we get that, since $X[i]$ is all zero in these rows and this also holds for $e_i$, that the vectors are a combination of these and thus also remain unchanged (since the altered modulo only applies to non-zero entries). 
    
    It remains to consider index $i$ for these vectors. Regarding vectors $C[i]$ we have shown in the proof of \autoref{thm:fast_euclid_correct} that the vectors are in the subspace of $B^{(i+1)}$. Hence, the $i$-th row of $X[i]$ is $0$ and remains unchanged. For $Z_j$ the $i$-th index is equal to $g_j / t = 1 / k$ for some integral $k$ as $g_j | t$ and $Z_j$ is calculated for the exact purpose of representing a vector on the translate, which is the gcd of the translates $t, t_1, \ldots, t_j$. Therefore, also the $i$-th entries remain unchanged and thus the lattice remains unchanged. 


\end{proof}

\subsection*{Size of the Solution Basis}
By Lemma~\ref{lem:modulo_entries} one can easily see that~\autoref{alg:fasteuclidean} returns a solution basis $S = BY$, with bounded entries. Recall that $B$ is the initially chosen matrix of maximum rank, then it holds that
\begin{align*}
    \norm{S} \leq d \norm{B} \leq d \norm{A}.
\end{align*}
This is because in~\autoref{alg:fasteuclidean}, each entry in the solution matrices $X$ and $Y$ is computed with the altered modulo. Since the matrices $X$ and $Y$ contain solution vectors, i.e. $B X_c = c$ and $B Y_i = S_i$ for some $c, S_i \in \lattice (A)$ this operation is similar to computing the respective vector $c$ or $S_i$ modulo $\Pi(B)$ but leaves vectors of $B$ untouched in case one of those are part of the solution basis.

In~\autoref{sec:root_n_solution_size} we present a postprocessing proceude which is applied the solution basis $S$ in order to obtain a basis $S'$ with a further improved size bound of $\norm{S'} \leq \sqrt{d}\norm{A}$.

\subsection{The Fractionality of the Parallelepipped} \label{sec:fractionality}
As one of the main operations used in~\autoref{alg:fasteuclidean}, we add and subtract 
 in l.9 and l.10 solution vectors $x \in \qq^\dd$.
In order to bound the bit complexity of~\autoref{alg:fasteuclidean} and explain the benefits of l.5 of the algorithm, we need a deeper understanding of the fractional representation of an integer point $b \in \zz^d$ for given basis $B \in \zz^{\dd \times \dd}$. Consider a solution of the linear system $x \in \qq^\dd$ with $Bx = b$. How large can the denominators of the fractional entries of $x$ get? We call the size of the denominators, the fractionality of $x$. Clearly by Cramer's rule, we know that the fractionality of $x$ can be bounded by $\det{B}$. However, there are instances where the fractionality of $x$ is much smaller than $\det{B}$. Consider for example the basis $B$ defined by vectors $B_i = 2 \cdot e_i$, where $e_i$ denotes the $i$-th unit vector. In this case, $\det{B}$ equals $2^\dd$. However, each integral point $b \in \zz^\dd$ can be represented by some $x\in \qq^\dd$ with denominators at most $2$.

For a formal notion of the fractionality, consider the parallelepipped $\Pi(B)$ for a basis $B \in \zz^{\dd \times j}$, we define the fractionality $\fract_B [i]$ of a variable $x_i$ to be the size of the denominator that is required to represented each integral point $b \in \Pi(B) \cap \zz^\dd$. More precisely, we define the fractionality as
\begin{align*}
    \fract_B [i] = \max_{b \in \zz^{\dd}} \{ z_i \mid x = B^{-1}b; x_i = \frac{y_i}{z_i}, \mathrm{gcd}(y_i,z_i) = 1 \}.
\end{align*}

The reason that some bases $B$ have a fractionality less than $\det(B)$ is that the faces in $\Pi(B)$ contain integral non-zero points. We show this relationship in the following lemma.\\

Exemplary, consider the parallelepiped for basis $B = \begin{pmatrix}
    6 &1 \\3 & 3
\end{pmatrix}$ in~\autoref{fig:translates_fractionality} with $\det(B)= 15$. The fractionality $\fract_B[2] = 5$, since the parallepepipped $\Pi(B \setminus B_2) = \Pi(B_1)$ contains $3$ integral points. On the other hand, $\fract_B [1] = 15$ as the subspace $\Pi(B \setminus B_1) = \Pi(B_2)$ contains only $0$ as an integral point.

\begin{figure}
    \hfill
    \subfloat[Translates of the Subspace $B \setminus B_2 = B_1$.\label{fig:translates_fractionality_1}]
    {\begin{tikzpicture}[scale=0.5]
 
 \def\X{6}
 \def\Y{6}
 
 
 
 \foreach \x in {0,...,\X}{
   \foreach \y in {0,...,\Y}
     \node[draw,inner sep=0,circle, fill] at (\x,\y) {};
 }
 

 \draw[-] (0,0) -- (6,3) node[below] {$B_1$};
 \draw[-] (0,0) -- (1,3) node[left] {$B_2$};
 
 \draw[-] (6,3) -- ++(1,3);
 \draw[-] (1,3) -- ++(6,3);

 \draw[dotted] (1/5,3/5) -- ++(6,3);
 \draw[dotted] (2/5,6/5) -- ++(6,3);
 \draw[dotted] (3/5,9/5) -- ++(6,3);
 \draw[dotted] (4/5,12/5) -- ++(6,3);
 

\end{tikzpicture}}
    \hfill
    \subfloat[Translates of the Subspace $B \setminus B_1 = B_2$.
    \label{fig:translates_fractionality_2}]
    {\begin{tikzpicture}[scale=0.5]
 
 \def\X{6}
 \def\Y{6}
 
 
 
 \foreach \x in {0,...,\X}{
   \foreach \y in {0,...,\Y}
     \node[draw,inner sep=0,circle, fill] at (\x,\y) {};
 }
 

 \draw[-] (0,0) -- (6,3) node[below] {$B_1$};
 \draw[-] (0,0) -- (1,3) node[left] {$B_2$};
 
 \draw[-] (6,3) -- ++(1,3);
 \draw[-] (1,3) -- ++(6,3);


    \foreach \x in {0,...,14}{
        \draw[dotted] (\x*6/15,\x*1/5) -- ++(1,3);
 }
 

\end{tikzpicture}}
    \hfill\quad
    \caption{Translates of the Subspaces}
    \label{fig:translates_fractionality}
\end{figure}
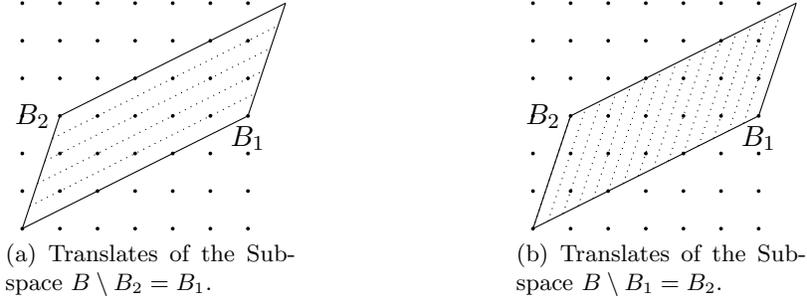

\begin{lemma}
    Let $B \in \zz^{\dd \times j}$ for some $j \leq \dd$ then
    \begin{align*}
        |\Pi(B) \cap \zz^\dd| = \fract_B [\ell] \cdot |\Pi(B \setminus B_\ell) \cap \zz^\dd|
    \end{align*}
    holds for every $\ell \leq j$.
\end{lemma}
\begin{proof}
    Let $L = |\Pi(B \setminus B_\ell) \cap \zz^\dd|$ and consider all points $p^{(1)}, \ldots , p^{(L)} \in \Pi(B \setminus B_\ell) \cap \zz^d$.
    
    Let $f \in (0,1]$ be the smallest number such that there exists a point $p' \in \Pi(B) \cap \zz^\dd$ with solution vector $x' \in \qq^j$ such that $B x' = p'$ and $x_\ell = f$. In the case that $f= 1$, all integral points $p \in \Pi(B) \cap \zz^\dd$ already belong to $\Pi(B \setminus B_\ell) \cap \zz^\dd$ and hence $\fract_B [\ell] = 1$ and the lemma holds.

    In the case that $f \in (0,1)$, observe first that $\frac{1}{f}$ must be integral. Otherwise, the point $\lceil \frac{1}{f} \rceil p \pmod {\Pi(B)} = B \{ \lceil \frac{1}{f} x' \rceil \}$ is a point in $\Pi \cap \zz^\dd$ with a smaller entry in the $\ell$-th component of the solution vector as $\lceil \frac{1}{f}  \rceil x'_{\ell} \pmod 1 < f$. Therefore, the parallelepipped $\Pi(B)$ contains exactly $1/f$ translates of the subparallepepipped $\Pi(B \setminus B_\ell)$ (including $\Pi(B \setminus B_\ell)$ itself). Consequently it holds that $\fract_B [\ell] = 1/f$.
    
    We will now show that each translate of $\Pi(B \setminus B_\ell)$ contains the same amount of integral points. For this, consider the points 
    \begin{align*}
        K \cdot p' + p^{(1)} \pmod {\Pi(B)}, \ldots , K \cdot p' + p^{(L)} \pmod {\Pi(B)}
    \end{align*} 
    \textbf{Observation 1}: Every point $K \cdot p' + p^{(i)}$ belongs to the translate $K$.\\
    The observation holds because $x^{(i)}_\ell = 0$ for every $x^{(i)}$ with $B x^{(i)} = p^{(i)}$.\\
    \textbf{Observation 2}: $K \cdot p' + p^{(i)} \pmod {\Pi(B)} \neq K \cdot p' + p^{(j)} \pmod {\Pi(B)}$ for $i \neq j$, i.e. all $L$ points on each translate are distinct.\\
    This observation holds as quality would imply that $p_i \equiv p_j \pmod {\Pi(B \setminus B_\ell)}$, which is a contradiction to the definition of the points.

    Conclusively, all integral points in $\Pi(B)$ are partitioned into exactly $1/f = \fract_B [\ell]$ translates each containing $L = |\Pi(B \setminus B_\ell) \cap \zz^\dd|$ integral points. This proves the Lemma.
\end{proof}
Iterating the above lemma over all subspaces yields the following corollary:
\begin{corollary}\label{cor:fractionality_sum}
Given is a (non-singular) basis $B \in \zz^{d\times d}$. Consider an arbitrary permutation $i_1, \ldots i_{\dd}$ of the indices $1, \ldots , \dd$ and let $C^{(j)}$ be the subspace defining matrix consisting of column vectors $B_{i_1} , \ldots , B_{i_j} \in \zz^d$, then
\begin{align*}
    |\det(B)| = |\Pi(B) \cap \zz^{\dd}| = \prod_{j=1}^{\dd} \fract_{C^{(j)}} [i_{j}].
\end{align*}
\end{corollary}

    

\subsection{The Bit Complexity of the fast Generalized Euclidean Algorithm} \label{sec:runtime_fast_euclid}
The main problem when analyzing the bit complexity of~\autoref{alg:fasteuclidean} is that it operates with solutions to linear systems. The solution matrices $X$ and $Y$ as defined in~\autoref{alg:fasteuclidean} however have entries with a potentially large fractionality. The fractionality can in general only be bounded by the determiant, which then can be bounded by Hadamard's inequality by $(d \norm{B})^d$. Therefore the bit representation of the fractioinal entries of the matrices have a potential size of $\Tilde{O}(\dd \cdot \log(\norm{A}))$. This would lead to a running time of $\Tilde{O}((n-d)\dd^3 \cdot \log(\norm{A}))$ for~\autoref{alg:fasteuclidean}.

However, in a more precise analysis of the running time of our algorithm, we manage to exploit the structural properties of the parallelepipped that we showed in the previous section. 
First, recall that in iteration $i$ of~\autoref{alg:fasteuclidean} the algorithm deals only with a solution matrix $X$ which contains solutions to a linear system $B^{(i)}x = c$, where $B^{(i)} $ is a matrix consisting of a subset of columns of matrix $B$.

In the previous section, we showed that the fractionality of solutions of the linear systems correlates with the number of integral points which are contained in the faces of the parallelepipped. The fractionality is low if there are many points contained in the respective faces $\Pi(B^{(i)})$.

Hence, on the one hand, if we have a basis which contains plenty of points in its respective subspaces $\Pi(B^{(i)})$ the algorithm is able to update the solution matrices $X$ and $Y$ rather efficiently as its fractionality must be small. On the other hand, if there are few points contained in the subspaces $\Pi(B^{(i)})$, updating $X$ and $Y$ requires many bit operations in the first iteration, however, since the number of points in the subspace is small, the fractionality of the subspaces is bounded. This makes the updating process in the following iterations very efficient.




\begin{theorem}\label{theo:bit_operations}
    The number of bit operations of~\autoref{alg:fasteuclidean} is bounded by
    \begin{align*}
        LS(\dd,n-\dd,\norm{A})  + \Tilde{O}((n-d)d^2 \cdot \log(\norm{A})),
    \end{align*}
    where $LS(\dd,n-\dd,\norm{A})$ is the time required to obtain an exact fractional solution matrix $X \in \qq^{\dd \times (n-\dd)}$ to a linear system $B X = C$, with matrix $B \in \zz^{\dd \times \dd}$ with $\norm{B} \leq \norm{A}$ and matrix $C \in \zz^{\dd \times (n-\dd)}$ 
\end{theorem}
\begin{proof}
    Throughout its execution the algorithm maintains a solution matrix $X \in \qq^{d \times (n-d)}$ of fractional entries. Equivalently, to the proof of~\autoref{thm:fast_euclid_correct}, let $X[i]\in \qq^{d \times (n-d)}$ be the solution matrix $X$ with values set as at l.10 of the $i$-th iteration and respectively $Z_j[i]$ be the vectors as defined in l.9 for every $j$ and iteration $i$. Let $C[i] \in \qq^{d \times (n-d)}$ be the matrix $C[i]$ as defined in (\ref{def:Ci}). Furthermore, we assume, that the algorithm pivots index $i$ in the $i$-th iteration and $B^{(i)} \in \zz^{d \times (n-\dd+1)}$ is the matrix containing columns $B_i, \ldots , B_\dd$. We define the fractionality of matrix $B^{(i)}$ to be the fractionality of the pivoted index $\ell = i$, i.e.
    \begin{align*}
        \fract(B^{(i)}) := \fract_{B^{(i)}}[i].
    \end{align*}
    For the matrices $X[i]$, $Y$, and vectors $Z_j[i]$, we maintain the property that each entry is of the form $\frac{a}{b}$ for some $a,b \in \zz$ with $gcd(a,b) = 1$ and $a < b$. By this we can ensure that $a,b < \fract(B^{(i)})$ in each entry of $X[i]$, $Y$, and $Z_j[i]$.
    
    For $X[0]$ being computed in l.3 of the algorithm, each numerator $a$ and denominator $b$ is bounded by the largest subdeterminant of $A$ and hence by Hadamard's inequality $a,b \leq d \norm{A}^d$. The property of the entries that $a<b$ and that $a,b$ being coprime can then be computed in time $\Tilde{O}((n-d)d^2 \log( \norm{A}))$ using Schönhage's algorithm~\cite{DBLP:journals/acta/Schonhage71} to compute the gcd.

    Consider now solution matrix $X[i-1]$ in the $i$-th iteration (before $X[i]$ is defined in l.10).
    In l.5 of the algorithm, the index $\ell$ with maximum fractionality in our current basis $B^{(i)} \in \zz^{d \times (n-\dd+1)}$ is being computed. To determine this index, we need to compute a common denominator for each row in $X[i-1]$.
    Hence, we compute for each column index $j$, the least common multiple (lcm) $L_j$ of all denominators in the $j$-th row of $X[i-1]$, i.e. the lcm of the denominator of the numbers $X_{jk}[i-1]$ for $1 \leq k \leq n-d$. Note that all $L_j$ can be computed in time $\tilde{O}((n-d)d \cdot \fract(B^{(i)}))$ using \cite{DBLP:journals/acta/Schonhage71} for every entry of the respective matrix.
    
    \textbf{Claim}: $L_j \leq \fract_{B^{(i)}}[j]$.\\
    The claim holds as in the solution vector $x$ with $Bx = p$ for point $p = \sum_{c \in C[i]} c \pmod {\Pi(B^{(i)})} \in \Pi(B^{(i)}) \cap \zz^d$, the $j$-th component $x_j$ equals $x_j = \sum_{k=1}^{n-d} X_{jk}[i] \pmod 1$ and hence can only be represented by a fractional number with denominator exactly $L_j$.


    Having computed the $L_i$, we can write the $i$-th component of each solution vector $X_c[i-1]$ by $\frac{a_c}{L_i}$, where $a_c$ is a divisor of the actual translate $t_c$ of vector $c \in C[i-1]$. In the same way, the translate $t$ of $B_i$ has to be a multiple of $L_i$. Therefore, the gcd computations $g_j = gcd(t,t_1, \ldots, t_j)$ in l.8 of the algorithm can be computed by $gcd(L_i,a_1, \ldots , a_j)$. Using the algorithm of Schönhage~\cite{DBLP:journals/acta/Schonhage71}, this can be done in time $\Tilde{O}((n-d) \cdot \log (\fract(B^{(i)}))$.
    
    For $X[i-1]$ with $i\geq1$ observe that each column vector $x$ in $X[i-1]$ is the solution vector of a linear system $B^{(i)} x = c$ for some $c \in C[i-1]$, the denominator $b \in \zz$ (and therefore also the numerator $a$) in each entry of $x$ is hence bounded by $\fract(B^{(i)})$. This implies that $Z_j[i]$ in l.9 and $X[i]$ in l.10 of the algorithm can be computed in time $\Tilde{O}((n-d)d \cdot \log(\fract(B^{(i)})))$.
    After the update of $X[i-1]$ in l.10 of the algorithm, the property of coprimeness of the rational numbers can be restored by computing the gcd of each entry in $X[i]$ and dividing accordingly. This requires a total time of $O((n-d)d \cdot \log(\fract(B^{(i)}))$ using the algorithm of Schönhage~\cite{DBLP:journals/acta/Schonhage71}.

    Summing up, the number of bit operations that are required over all $d$ iterations of the for-loop in~\autoref{alg:fasteuclidean} is bounded by
    \begin{align*}
        \Tilde{O}((n-d)d^2 \log( \norm{A})) + \sum_{i=1}^\dd \Tilde{O}((n-d)d \cdot \log(\fract(B^{(i)}))) \\
        = \Tilde{O}((n-d)d^2 \log( \norm{A})) + (n-d)d\sum_{i=1}^\dd \Tilde{O}(\log(\fract(B^{(i)}))).
    \end{align*}
    By~\autoref{cor:fractionality_sum}, we know that $\prod_{i=1}^d \fract(B^{(i)}) = |\Pi(B) \cap \zz^\dd| = \det(B)$ and hence the sum in the above term can be bounded by
    \begin{align*}
        \sum_{i=1}^\dd \Tilde{O}(\log(\fract(B^{(i)}))) = \Tilde{O}(\log(\det(B))) = \Tilde{O}(d \log (\norm{A})).
    \end{align*}
    Therefore, the running time required over all iterations in the for-loop of the algorithm is bounded by
    \begin{align*}
        \Tilde{O}((n-d)d^2 \log( \norm{A})).
    \end{align*}
    Furthermore, it holds that the matrix multiplication in l.12 of~\autoref{alg:fasteuclidean} can be upper bounded by $\tilde{O}(d^{\omega}\norm{A})$ using~\autoref{lem:matrix_mult} from the following section. This is because $Y_i \in [0,1)^d$, fractional, and with denominator dividing $\fract(B^{(i)})$. \autoref{lem:matrix_mult} computes the matrix multiplication for two \emph{integral} matrices. Therefore, we compute $Y_i' := \fract(B^{(i)}) \cdot Y_i \in \zz^d$ and get that $\norm{Y_i'} \leq \fract(B^{(i)})^2$. Clearly, this can be done in target complexity. Then we use \autoref{lem:matrix_mult} to compute $BY'$ with $k=1$. Again by \autoref{cor:fractionality_sum} we get that 
    \begin{align*}
        \nu := \sum_{i=1}^{d}\log(\norm{Y_i'}  +1) / \log\norm{B}   
        &\leq \sum_{i=1}^{d}\log(\fract(B^{(i)})^2 +1) / \log\norm{B} \\ 
        &\leq O(\log\det B)/ \log\norm{B} \\
        &=O(d\log(d)).
    \end{align*}
    The running time by the lemma is $\tilde{O}(\max\{\nu, d\} d^{\omega -1} \norm{B} ) = \tilde{O}(d^\omega\norm{B} )$. Finally, we divide each column again by  $\fract(B^{(i)})$ to obtain $S = BY$. Since in any case $LS(\dd,n-\dd,\norm{A}) \in \Omega(d^{\omega} \cdot \log \norm{A})$, we obtain that the running time for l.12 is swallowed by the linear system solving in l.3.
\end{proof}

We want to repeat a fact from the proof which we will need again later.
\begin{observation}\label{ob:bounded_representation_of_solution}
    The matrix $Y$ from l.11 has bounded entries as numerator and denominator are bounded by $\fract(B^{(i)})$ and with \autoref{cor:fractionality_sum} the sum of the encoding lengths of the maximum number in each column is bounded by
    \[\sum_{i=1}^\dd \Tilde{O}(\log(\fract(B^{(i)}))) = \Tilde{O}(d \log (\norm{A})).\]
\end{observation}

\subsection*{Arithmetic Complexity}

We want to give a short overview of the arithmetic complexity. Lines 1 to 3 require $O(nd^{\omega -1})$ operations. Line 3 in particular can be computed by first computing $B^{-1}$ and then $X= B^{-1}C$. Roughly speaking, lines 5 to 8 are dominated by computing $O((n-d)d)$ gcd-like computations. While these computations might involve exponentially large numbers, we have seen in the proof of~\autoref{theo:bit_operations} that this can only happen in few iterations. By arguments as above the algorithm requires $\OTilde((n-d)d\log\norm A_\infty)$ arithmetic operations for lines 5 to 8. Lines 9 to 11 require $O((n-d)d)$ arithmetic operations and line 12 is only one matrix multiplication. We get the following bound on the arithmetic complexity.

\begin{corollary}
    \autoref{alg:fasteuclidean} computes a lattice basis of $A$ using $\OTilde(nd^{\omega - 1} + (n-d)d^2\log\norm{A}_\infty)$ arithmetic operations.
\end{corollary}

\subsection{Matrix Multiplication}\label{sec:matrix_mult}

In the previous section we analyzed the bit complexity of \autoref{alg:fasteuclidean} including the matrix multiplication $S=BY$ in l.12. In a naive approach entries of $Y$ might be as large as $\det B \leq (d\norm{B})^d$ and thus the running time would increase by a factor of $d$. In this section we use a technique similar to~\cite[Lemma 2]{DBLP:conf/issac/BirmpilisLS19} that analyses the complexity of matrix multiplication based on the dimension and the size of coefficients in each column of the second matrix. Compared to~\cite[Lemma 2]{DBLP:conf/issac/BirmpilisLS19}, there are two main differences in our analysis. First, we consider the magnitude of each column individually instead of the magnitude of the entire matrix. Second, we allow rectangular matrix multiplication in order to improve the running time in the case that there are many columns or very large numbers in the second matrix. 

The following lemma both improves the running time of the matrix multiplication in l.12 significantly compared to the naive approach as well as it is a key component in the following section for analyzing the complexity of solving a linear system $BX=C$.

\begin{lemma}\label{lem:matrix_mult}
    Consider two matrices $M\in \zz^{a \times a}$ and $N\in \zz^{a \times b}$ and let $\nu := \sum_{i=1}^{b}\log(\norm{N_i} +1) / \log \norm{M} $. Consider any $k\geq 1$. Then the matrix multiplication $M\cdot N$ can be performed in
    \[\tilde{O}(\max\{\nu, a^k\} \cdot a^{\omega(k)-k}\log \norm{M} )\]
    bit operations.
\end{lemma}
\begin{proof}
    Consider any $k\geq 1$. Without loss of generality, we will first consider the case that $M$ and $N$ both consist of non-negative entries only. 
    
    Let $X \in \zz_{>0}$ be the smallest power of $2$ such that $X > \norm{M} $. Let $N$ have the $X$-adic expansion $N = \sum_{i=0}^{p-1}X^iN^{(i)}$ for some $p\in \zz$ and set 
    \begin{align*}
        N' = [\ N^{(0)}\ |\ N^{(1)}\ |\ \ldots\ |\ N^{(p-1)}\ ] \in \zz^{a \times bp}.
    \end{align*}
    Although at first thought the matrix has $bp$ columns, the size of entries for each column $i$ of $N$ decides the size of the $X$-adic expansion of that column $i$ individually. Thus, there are at most $\sum_{i=1}^{b}\log_{X}(\norm{N_i} +1) = O(\nu)$ many non-zero columns. Compute the matrix product
    \begin{align*}
        MN' = [\ MN^{(0)}\ |\ MN^{(1)}\ |\ \ldots\ |\ MN^{(p-1)}\ ]
    \end{align*}
     using rectangular matrix multiplication for dimensions $a \times a$ and $a\times a^k$. We add zero columns such that the number of columns in $N'$ is a multiple of $a^k$. Then compute $O(\nu / a^k)$ rectangular matrix multiplications in order to obtain $MN'$. This can be done in target time as each of the $O(\nu / a^k)$ matrix multiplications costs $\tilde{O}(a^{\omega(k)}\log X)$ bit operations, which results in time $\tilde{O}(\max\{\nu, a^k\} \cdot a^{\omega(k)-k}\log \norm{M} )$.

    Now compute $MN = \sum_{i=0}^{p-1}X^iMN^{(i)}$. The sum can be computed in $\tilde{O}(\nu a\norm{M} )$ bit operations since augmenting $\sum_{i=0}^{\ell}X^iMN^{(i)}$ to $\sum_{i=0}^{\ell+1}X^iMN^{(i)}$ for some $\ell < p-1$ only requires changes in the leading $d$ bits for some $d \in O(\log(a X)) = O(\log(a\norm{M} ))$ and in total over the entire sum only $\nu$ columns are added (as others will be zero).
    
    Let us now consider the case of matrices containing also negative entries. Define $M^{(+)}$ as the matrix $M$ but all negative entries of $M$ are replaced by $0$ and let $M^{(-)} := M^{(+)} - M$. Define $N^{(+)}$ and $N^{(-)}$ similarly with respect to $N$. Thus, we get that $M = M^{(+)} - M^{(-)}$ and  $N = N^{(+)} - N^{(-)}$. Furthermore, all entries in $M^{(+)}$, $M^{(-)}$, $N^{(+)}$, and $N^{(-)}$ are non-negative. Using the procedure above we compute $M^{(+)}N^{(+)}$, $M^{(+)}N^{(-)}$, $M^{(-)}N^{(+)}$, and $M^{(-)}N^{(-)}$ in target time. We then get the result by computing 
    \begin{align*}
        MN = (M^{(+)} - M^{(-)})(N^{(+)} - N^{(-)}) = M^{(+)}N^{(+)} - M^{(+)}N^{(-)}-M^{(-)}N^{(+)} + M^{(-)}N^{(-)}.
    \end{align*}
\end{proof}

As rectangular matrix multiplication requires the same running time for dimensions $a\times a$ and $a\times a^k$, $a^k\times a$ and $a \times a$, and $a\times a^k$ and $a^k \times a$, small tweaks of the algorithm above might also solve similar special cases. However, as we only require matrix multiplications of this form, we will not discuss this further. 

Using $\nu = \sum_{i=1}^{b}\log(\norm{N_i} +1) / \log \norm{M}  \in  O(b\cdot \log\norm{N}  / \log \norm{M})$ we get the following simpler statement, which is already sufficient in many cases. The precise version above improves the running time only if entry sizes of columns from the second matrix differ exponentially or the analysis of the magnitude spreads the complexity over all columns. In our case this applies to l.12 of \autoref{alg:fasteuclidean}, where we compute the solution matrix $S=BY$.

\begin{corollary} \label{cor:matrix_mult}
    Consider two matrices $M\in \zz^{a \times a}$ and $N\in \zz^{a \times b}$ and any $k\geq 1$. Then the matrix multiplication $M\cdot N$ can be performed in 
    \begin{align*}
        \tilde{O}(\max\{b\cdot \log\norm{N}/\log\norm{M}, a^k\}a^{\omega(k)-k}\log\norm{M}).
    \end{align*}
\end{corollary}

\subsection{Linear System Solving} \label{sec:lin_sys_solving}
The main complexity of~\autoref{alg:fasteuclidean} stems from solving the linear system $BX=C$ in l.3. In this section we combine the algorithm \texttt{solve} from~\cite{DBLP:conf/issac/BirmpilisLS19} with our \autoref{lem:matrix_mult} for matrix multiplication in the previous lemma. Our analysis considers linear systems with a matrix right-hand side instead of a vector right-hand side.

\begin{lemma}\label{lem:lin_sys_solving}
    Given a nonsingular matrix $M \in \zz^{a \times a}$ and a right hand-side matrix $R \in \zz^{a \times b}$, where $\log\norm{R} \in O(a\log\norm{M})$.
    Consider any $k > 0$. Then 
    \begin{align*}
        LS(a,b,\norm{M}, \norm{R}) \in \tilde{O}(\max\{b, a^k\}a^{\omega(k+1)-k}\log\norm{M}).
    \end{align*}
\end{lemma}
\begin{proof}

Consider any $k>0$. We analyze algorithm \texttt{solve} from Birmpilis et.~al~\cite[Figure~8]{DBLP:conf/issac/BirmpilisLS19}, see \autoref{alg:solve},  for a right-hand side matrix instead of a vector. 
\begin{algorithm}[h!]
    \caption{Solve (Linear System)}
    \label{alg:solve}
  \begin{algorithmic}[1]
    \Statex{
    \textsc{Input: } Nonsingular $M\in \zz^{a\times a}$ and $R\in \zz^{a\times b}$.}
    \Statex{
    \textsc{Output: } $X\in \zz^{a \times b}$ and $e \in \zz_{\geq 0}$ such that $e$ is minimal such that all denominators of the entries in $2^e M^{-1}$ are relatively prime to $2$, and $x = \mathrm{Rem}(2^e M^{-1}R, 2^d )$ where $d$ is as defined in step $3$.}
    \Statex{
    \textsc{Note: } $2^e M^{-1} R = \mathrm{RatRecon}(x,2^d,N,D)$.}
    \State{$(P,S,Q)) := \mathrm{Massager}(M,a)$ }
    \Statex{$e := \log_2S_{aa}$}
    \State{$U := \mathrm{ApplyMassager}(M,a,P,S,Q,a)$}
    \State{$N := \floor{a^{a/2}\norm{M}^{a-1}\norm{R}}$}
    \Statex{$D := \floor{a^{a/2}\norm{M}^{a} / 2^e} $}
    \Statex{$d := \lceil \log(2ND)\rceil $}
    \Statex{$Y := \mathrm{SpecialSolve}(U,R,d,a,b)$}
    \State{$X := \mathrm{Rem}(PQ(2^eS^{-1})Y, 2^d)$}
    \State \Return{$X,e$}
  \end{algorithmic}
\end{algorithm}

Correctness follows directly from their proof. 
As most of their analysis directly caries over and changes  appear in matrix dimensions only, our analysis will focus on the differences required for our running time. Throughout their paper they phrase the results under the condition of a \emph{dimension $\times$ precision invariant} that restricts part of the input to be near-linear in $a\log\norm{M}$. For a right-hand side matrix this invariant might no longer be satisfied. So in other words, we instead parameterize over the size of this quantity by a polynomial $k$ using rectangular matrix multiplication $\omega(k+1)$ from \autoref{lem:matrix_mult}. 

Compared to the analysis from~\cite{DBLP:conf/issac/BirmpilisLS19}, we need to take a closer look into two calculations from the algorithm. In the algorithm steps $1$ and $2$ do not depend on the right-hand side and thus have the same running time as before, which is $\tilde{O}(a^{\omega}\log\norm{M})$.  From step $3$ computing $N$, $D$, and $d$ is also in target time. 

Thus, the first thing we need to analyze is $Y := \mathrm{SpecialSolve}(U,R,d,a,b)$ in step 3. 
For a swift analysis, let us analyze some magnitudes from the algorithm. We get that 
\[d \in O(\log (a^a\norm{M}^{2a}\norm{R})) = \tilde{O}(a\log\norm{M}).\] 
Numbers involved are bounded by $\chi^{2^{\ell+1}+1}$ with $\chi \in O(a^2\norm{M}))$ and $\ell \leq \log(d+1)$. Thus, we get that numbers involved are at most $ O(a^2\norm{M}))^{2^{\log( d +1)+1}+1} = 2^{\tilde{O}(a\log\norm{M})}$.

Their analysis of \texttt{SpecialSolve} requires the dimension $\times$ precision invariant $b\cdot d \in O(a\log(a \norm{M}))$, which is not necessarily the case here. 
However, the running time is dominated by $\ell \in \tilde{O}(\log(a\log\norm{M}))$ matrix multiplications of an 
$a\times a$ matrix  with coefficients of magnitude $O(a^2\norm{M})$ and an $a\times b$ matrix with coefficients of magnitude $2^{\tilde{O}(a\log\norm{M})}$. 

We now want to perform the matrix multiplication using \autoref{cor:matrix_mult}. We fill the second matrix with columns of zeroes such that it is a multiple of $a^{k}$, say $b'$ columns. Using \autoref{cor:matrix_mult} with $k' = k+1$, we perform each matrix multiplication in time 
\begin{align*}
    &\tilde{O}(\max\{b' \cdot a\log\norm{M} / \log\norm{M}, a^{k'}\} a^{\omega(k')-k'}\cdot a \log \norm{M}) \\&= \tilde{O}(\max\{b,a^k\}a^{\omega(k+1)-k}\log\norm{M}).
\end{align*} 

The second calculation we need to analyze is $\rem(PQ(2^eS^{-1})Y, 2^d)$ from step 4. The first part $Z := (2^eS^{-1})Y$ involves a diagonal matrix $S^{-1}$ and can be computed in time. For the multiplication $QZ$, we again roughly follow the steps from their paper. By their Lemma 17, the $\chi'$-adic expansion of the columns of $Q$ consists of $a' \leq 2a$ columns for $\chi'$ the smallest power of $2$ such that $\chi'\geq \sqrt{a}\norm{M}$. Let $Q' = \left(Q_0 \ldots Q_{p-1} \right)$ be the $\chi'$-adic expansion of $Q$, where $Q_i\in\zz^{a\times k_i}$ and $\sum_{i <p}k_i = a' \leq 2a$. Let $ Z = \left( Z_0 \ldots Z_{p-1} \right)$ be the $\chi'$-adic expansions of $Z$ and let $Z_i^{(k_i)}$ be the submatrix of the last $k_i$ rows. The matrix multiplication can be restored from the product
\begin{align}
    \begin{pmatrix}Q_0 \ldots Q_{p-1}\end{pmatrix}
    \begin{pmatrix}
    Z_0^{(k_0)} & Z_1^{(k_0)} & \ldots & Z_{p-1}^{(k_0)} \\
    & Z_0^{(k_1)} & \ldots & Z_{p-2}^{(k_1)} \\
    && \ddots & \vdots \\
    &&& Z_0^{(k_{p-1})}
    \end{pmatrix}.
\end{align}
The dimensions are $a \times a'$ and $a' \times bd$ since coefficients of $Z$ are bounded by $2^d$. Hence, using \autoref{cor:matrix_mult} the matrix multiplication $QZ$ can also be computed in time $\tilde{O}(\max\{b,a^k\}a^{\omega(k+1)-k}\log\norm{M})$, similar to the matrix multiplications above. Note that $a' \neq a$ is not a problem as $a' \leq 2a$ and for example filling up rows of zeroes and calculating with an $a'\times a'$ matrix lets us apply the corollary directly. 

Finally, we apply the permutation matrix $P$ to $QZ$ and obtain the result $X$ in target time.  
\end{proof}

Inserting \autoref{lem:lin_sys_solving} for linear system solving into \autoref{theo:bit_operations}, we arrive at our final result for the bit complexity of \autoref{alg:fasteuclidean}.
\begin{corollary}
    The number of bit operations of \autoref{alg:fasteuclidean} is bounded by
    \[\tilde{O}(\max\{n-d, d^k\}d^{\omega(k+1)-k}\log\norm{A})\quad \text{for any }k\geq 0.\]
\end{corollary}

Note that for $n-d\in o(d)$ we choose $k = 0$ and get $\tilde{O}((n-d)d^{\omega}\log\norm{A})$.  In the case of $n-d \in O(1)$ this matches the running time of the algorithm  from Li and Storjohann~\cite{DBLP:conf/issac/LiS22}, which is restricted to this special case. We get the smallest improvement over the currently fastest general algorithm by Storjohann and Labahn~\cite{DBLP:conf/issac/StorjohannL96} in the regime of $n-d \in \Theta(d)$. In that regime we choose $k=1$ and get a running time of $\tilde{O}(d^{\omega(2)}\log\norm{A})$. The improvement for current values of $\omega$ and $\omega(2)$ is a factor of $d^{\omega - \omega(2) + 1} \approx d^{0.121356}$~\cite{williams2024new_matrixmultiplication}. Considering larger amounts of additional vectors, also our improvement over the algorithm by Storjohann and Labahn~~\cite{DBLP:conf/issac/StorjohannL96} gets slightly better. Say for example $n-d \in \Theta(d^4)$, then we choose $k=4$ and get a running time of $\tilde{O}((n-d)d^{\omega(5)-4}\log\norm{A})$. This is an improvement of a factor of $d^{\omega - \omega(5) + 4} \approx d^{0.215211}$ for current values of $\omega$ and $\omega(5)$~\cite{williams2024new_matrixmultiplication}. 

\subsection{Lower Rank Lattices}\label{sec:lower_dimensions}

Our algorithm~\ref{alg:fasteuclidean} can be easily adapted to also work for lattice inputs of lower rank. Instead of finding $d$ linearly independent vectors, we find a full rank submatrix $\hat{B}$. Then we proceed as before but restricted to the rows provided by $\hat{B}$. As $Y$ is the representation of the resulting basis as a combination of vectors in $\hat{B}$, a multiplication with $B$ restores the result in all $d$ dimensions. The adjusted algorithm is as follows. 

\begin{algorithm}
  \caption{Fast Generalized Euclidean Algorithm (for lower rank lattices)
    \label{alg:fasteuclideanforlowerdimensions}}
  \begin{algorithmic}[1]
    \Statex{
    \textsc{Input: } A matrix $A=\left(A_1,\ldots, A_n \right) \in \zz^{\dd\times n}$ of rank $\hat{d} \leq d$.}
    \State{\textbf{find} independent vectors $B := \left(B_1, \ldots, B_{\hat{d}}\right)$ with $B_i\in \{A_1,\ldots, A_n\}$}
    \State{\textbf{find} independent rows $k_1, \ldots, k_{\hat{d}}$ of $B$ and let $\pi$ be a projection on these (row) indices}
    \State{\textbf{let} $C$ be a matrix with columns $\{A_1,\ldots, A_n\}\setminus \{B_1,\ldots, B_{\hat{d}}\}$}
    \State{\textbf{let} $\hat{C} = \pi(C)$ and $\hat{B} = \pi(B)$}
    \State{\textbf{solve} $\hat{B} X = \hat{C}$}
    \For{$i=1$ to $\hat{d}$}
        \State{\textbf{choose} index $\ell$ with maximum fractionality in matrix $X$}
        \State{\textbf{determine} the translate $\notd$ of $\hat{B}_\ell$ and the translate $t_i$ of vectors $\hat{C} = (\hat{C}_1, \ldots , \hat{C}_{n-d})$}
       \Statex{\hspace{2.65cm}with respect to the subspace $\hat{B} \setminus \hat{B}_\ell$}
        \For{$j=1$ to $n-\hat{d}$} 
        \State{\textbf{compute} $g_j := gcd(t, t_1, \ldots , t_j)$ and factors $\alpha_j, \beta_j$ such that $g_j = \alpha_j g_{j-1} + \beta_j t_j$} 
       \State{\textbf{set} $Z_j = \alpha_j Z_{j-1} + \beta_j X_{\hat{C}_j}$ with $Z_0 = \pi(e_{\ell})$}
       \State{\textbf{set} $X_{\hat{C}_j} = g_{j-1}/g_j \cdot X_{\hat{C}_{j}} - t_{j}/g_j \cdot Z_{j-1}$}
       \EndFor
       \State{\textbf{set} $Y_i = Z_{n-\hat{d}}$}
    \EndFor
    \State \Return{Basis $S = B Y$, with matrix $Y$ containing $Y_i$ as the $i$-th column}
  \end{algorithmic}
\end{algorithm}

\begin{theorem}
    \autoref{alg:fasteuclideanforlowerdimensions} computes a lattice basis of $A\in\zz^{d \times n}$ with $\hat{d} := rank(A)$ using bit complexity of 
    \begin{align*}
         \Tilde{O}(\max\{n,d\}d^{\omega-1}\log\norm{A}) + LS(\hat{\dd},n-\hat{\dd},\norm{A})  + \Tilde{O}((n-\hat{d})\hat{d}^2 \cdot \log(\norm{A})).
    \end{align*}
\end{theorem}
\begin{proof}
    In order to proof the correctness of~\autoref{alg:fasteuclideanforlowerdimensions}, we need two observations. The first observation is that although $Y$ might be fractional, every $Y_i$ describes an integral combination of vectors from $A$. Now this implies that we also get integral values in the components that are left out in the projection $\pi$. The second observation is that the lattice of $S$ actually captures all elements of the lattice of $A$. With arguments similar to \autoref{thm:fast_euclid_correct} we get correctness in the subspace $\pi(\lattice(S))= \pi(\lattice(A))$. The lattice of $A$ has rank $\hat{d}$. Hence, as there are no vectors in $\lattice(A)$ but not in $\lattice(S)$ considering $\hat{d}$-dimensional subspace of $\pi$, we get $\lattice(A) = \lattice(S)$.

    The running time follows similar to \autoref{theo:bit_operations} by adjusting the analysis to dimension $\hat d$ for all operations in l.6-13. 
\end{proof}

\section{Reducing the Size of the Solution Basis}\label{sec:root_n_solution_size}
Our presented algorithm~\ref{alg:fasteuclidean} already guarantees a linear size bound for the resulting basis of $\norm{S} \leq d\norm{A}$. In this section, we present a post processing procedure which modifies the returned basis matrix $S$ to achieve an improved size bound. We use a well-known method based on a probabilistic analysis from discrepancy theory. Essentially, this is a simple trick that applies the linearity of expectation to balance vectors. See also Alon and Spencer~\cite[section 2.6]{alon2000probabilistic}.

As the argument is elegant and for self-containment we want to give a brief overview of the idea. Given vectors $v^{(1)}, \ldots v^{(n)} \in \mathbb{R}^d$ and a vector $x \in [0,1]^n$ the task of balancing vectors is to find a vector $y \in \{0,1\}^n$ such that $\norm{\sum_{j=1}^n (x_j -y_j)v^{(j)}}$ is small. The greedy algorithm simply chooses each $y_i$ such that the partial sum $w^{(i)} = (x_1 - y_1)v^{(1)} + \ldots + (x_i - y_i)v^{(i)}$ has minimal norm. In order to sketch the resulting bound, consider for a moment $y_i$ to be a random variable with $\mathrm{Pr}[y_i = 1]=x_i$. With linearity of expectation, we get that 
\begin{align*}
    E[\norm{w_i}_2^2] &= E[\sum_{j=1}^n (w^{(i-1)}_j + (x_{i} - y_{i})v^{(i)}_j)^2]\\
    &= \norm{w^{(i-1)}}^2 + 2\norm{w^{(i-1)}\cdot v^{(i)}}^2 E[x_i - y_i] + \norm{v^{(i)}}^2E[(x_i - y_i)^2]\\
    &= \norm{w^{(i-1)}}^2 + x_i(1-x_i)\norm{v^{(i)}}^2.
\end{align*}
The expectation implies that there also exists a choice of $y_i \in\{0,1\}$ such that the expectation is at least met. At the end we get $\norm{w_n}^2 \leq \sum_{j=1}^n x_i(1-x_i)\norm{v^{(i)}} \leq n/4 \cdot\max_{i\leq n}\norm{v_i}^2$.

The procedure to reduce lattice bases with this method applies only to lower triangular matrices ($Y$ in our algorithm) with a specific structure in the diagonal entries, and the running time is again enhanced by our structural result on fractionality. Without leveraging this structural property, the algorithm would require an additional factor of $d$ in the running time. In the following we denote with $\norm{M}_2$ the maximum euclidean norm of any column vector of a matrix $M$.
\begin{theorem}
    Given $B$ and $Y$ as in \autoref{alg:fasteuclidean} l.12, we can compute a matrix $S \in \zz^{d\times d}$ such that
    \begin{itemize}
        \item $\lattice(S) = \lattice(BY)$ and
        \item $\norm{S}_2 \leq \max\{1,\frac{\sqrt{d}}{2}\}\norm{B}_2$
    \end{itemize}
    using $\OTilde(d^3\log\norm{B})$ bit operations.
\end{theorem}
\begin{proof}
     We compute each column of $S$ separately. Therefore, consider any column $i\leq n$. We consider two cases. The first case is that $Y_{ii} = 1$. In this case we set $Y'_{i} = e_i$, \ie{} the $i$-th unit vector, which clearly satisfies the size bound. 

    In the second case we get that $Y_{ii} \leq 1/2$.\footnote{In fact it is $1/ \fract_B[i]$. This follows from the gcd operation on the translates. The structure $1/ z$ for $z\in \zz$ follows from l.8 as the translate $t$ of $B_\ell$ is also the common denominator which was computed as the lcm of simplified fractions. Hence their gcd is $1$. } We use the greedy algorithm as described above from~\cite[section 2.6]{alon2000probabilistic} to compute a suitable vector $E[i] \in \{0,1\}^d$.\footnote{Note that we do not need to compute any square root although the algorithm formally uses the euclidean norm. In fact, we only need the norm squared.} In this case set $Y'_i = Y_i - E[i]$. 
    
    As described above, for each column $Y'_i$ we only need $O(d)$ calculations of the norm which results in $O(d^2)$ arithmetic operations in order to find $E[i] \in \{0,1\}^d$ such that $\norm{B(Y_i - E[i])}_2 \leq \sqrt{d}/2\norm{B}_2$. Due to our structural result,~\autoref{cor:fractionality_sum}, the running time does not blow up considering bit complexity but aligns with the arithmetic complexity up to logarithmic factors even if in the worst case this is done for every $i\leq d$. This follows from a bound on the numerator and denominator $Y$ as repeated in~\autoref{ob:bounded_representation_of_solution}. Thus, over all vectors, we require $\OTilde(d^3\log\norm{B}_\infty)$ bit operations.  
    
    We set $S:= BY'$. Next, we want to prove that also $\lattice(S) = \lattice(BY)$ is satisfied.
    Note that $Y$ is a lower triangular matrix, thus for column $i$ it is sufficient to consider (row) indices $i, \ldots, d$.    
    For the property $\lattice(S) = \lattice(BY)$ it is very important that our operations above do not change the translate, or in other words $Y_{ii}$ remains the same. In the first case it is unchanged as $Y_{ii} = Y'_{ii} = 1$. In the second case the greedy algorithm chooses the next index $E[i]_j$ from $\{0,1\}$ minimizing the norm of the partial sum.  From the triangular structure of $Y$ and the fact that $Y_{ii} \leq 1/2$, we get that $E[i]_i = 0$. In other words, $Y'_{ii} = Y_{ii} - E[i]_i = Y_{ii}$ and hence is unchanged. Clearly, the calculated vectors $BY'$ are in the lattice $\lattice(BY)$ since they are either a column vector from $B$ or stem from $BY$ but some vectors of $B$ might be subtracted (which is fine since $B$ is also part of the lattice). Thus we get that $\lattice(BY') \subseteq \lattice(BY)$.  Moreover, as we did not change the diagonal entries of $Y$ and $Y$ is a lower triangular matrix, we get that $\det(Y) =  \det(Y')$ and thus $\det(BY) = \det(BY')$. Hence, we get $\lattice(S) = \lattice(BY)$.
\end{proof}

\bibliographystyle{alpha}
\bibliography{references}

\begin{thebibliography}{WXXZ24}

\bibitem[Ajt96]{DBLP:conf/stoc/Ajtai96}
Mikl{\'{o}}s Ajtai.
\newblock Generating hard instances of lattice problems (extended abstract).
\newblock In Gary~L. Miller, editor, {\em {ACM} Symposium on the Theory of Computing}, pages 99--108. {ACM}, 1996.

\bibitem[AS00]{alon2000probabilistic}
N.~Alon and J.H. Spencer.
\newblock {\em The Probabilistic Method}.
\newblock Wiley Series in Discrete Mathematics and Optimization. Wiley, 2000.

\bibitem[BLS19]{DBLP:conf/issac/BirmpilisLS19}
Stavros Birmpilis, George Labahn, and Arne Storjohann.
\newblock Deterministic reduction of integer nonsingular linear system solving to matrix multiplication.
\newblock In {\em {ISSAC} 2019}, pages 58--65. {ACM}, 2019.

\bibitem[BLS23]{DBLP:journals/talg/BirmpilisLS23}
Stavros Birmpilis, George Labahn, and Arne Storjohann.
\newblock A cubic algorithm for computing the hermite normal form of a nonsingular integer matrix.
\newblock {\em {ACM} Trans. Algorithms}, 19(4):37:1--37:36, 2023.

\bibitem[BP87]{DBLP:conf/eurocal/BuchmannP87}
Johannes Buchmann and Michael Pohst.
\newblock Computing a lattice basis from a system of generating vectors.
\newblock In {\em {EUROCAL} '87}, volume 378 of {\em Lecture Notes in Computer Science}, pages 54--63. Springer, 1987.

\bibitem[CC82]{DBLP:journals/siamcomp/ChouC82}
Tsu{-}Wu~J. Chou and George~E. Collins.
\newblock Algorithms for the solution of systems of linear diophantine equations.
\newblock {\em {SIAM} J. Comput.}, 11(4):687--708, 1982.

\bibitem[CN97]{DBLP:conf/focs/CaiN97}
Jin{-}yi Cai and Ajay Nerurkar.
\newblock An improved worst-case to average-case connection for lattice problems.
\newblock In {\em {FOCS}}, pages 468--477. {IEEE} Computer Society, 1997.

\bibitem[Fru77]{DBLP:conf/fct/Frumkin77}
Michael~A. Frumkin.
\newblock Polynomial time algorithms in the theory of linear diophantine equations.
\newblock In {\em Fundamentals of Computation Theory}, volume~56 of {\em Lecture Notes in Computer Science}, pages 386--392. Springer, 1977.

\bibitem[Gal24]{gall2024faster_rectangularMM}
Fran{\c{c}}ois~Le Gall.
\newblock Faster rectangular matrix multiplication by combination loss analysis.
\newblock In {\em Proceedings of the 2024 Annual ACM-SIAM Symposium on Discrete Algorithms (SODA)}, pages 3765--3791. SIAM, 2024.

\bibitem[GPV08]{DBLP:conf/stoc/GentryPV08}
Craig Gentry, Chris Peikert, and Vinod Vaikuntanathan.
\newblock Trapdoors for hard lattices and new cryptographic constructions.
\newblock In {\em {ACM} Symposium on Theory of Computing}, pages 197--206. {ACM}, 2008.

\bibitem[HM91]{DBLP:journals/siamcomp/HafnerM91}
James~L. Hafner and Kevin~S. McCurley.
\newblock Asymptotically fast triangularization of matrices over rings.
\newblock {\em {SIAM} J. Comput.}, pages 1068--1083, 1991.

\bibitem[HPS11]{DBLP:conf/crypto/HanrotPS11}
Guillaume Hanrot, Xavier Pujol, and Damien Stehl{\'{e}}.
\newblock Analyzing blockwise lattice algorithms using dynamical systems.
\newblock In {\em {CRYPTO}}, volume 6841 of {\em Lecture Notes in Computer Science}, pages 447--464. Springer, 2011.

\bibitem[Ili89]{DBLP:journals/siamcomp/Iliopoulos89}
Costas~S. Iliopoulos.
\newblock Worst-case complexity bounds on algorithms for computing the canonical structure of finite abelian groups and the hermite and smith normal forms of an integer matrix.
\newblock {\em {SIAM} J. Comput.}, 18(4):658--669, 1989.

\bibitem[KB79]{DBLP:journals/siamcomp/KannanB79}
Ravindran Kannan and Achim Bachem.
\newblock Polynomial algorithms for computing the smith and hermite normal forms of an integer matrix.
\newblock {\em {SIAM} J. Comput.}, pages 499--507, 1979.

\bibitem[KR23]{klein2023old_paper}
Kim-Manuel Klein and Janina Reuter.
\newblock Simple lattice basis computation -- the generalization of the euclidean algorithm, 2023.

\bibitem[LLL82]{lll_algorithm}
Arjen~K Lenstra, Hendrik~Willem Lenstra, and L{\'a}szl{\'o} Lov{\'a}sz.
\newblock Factoring polynomials with rational coefficients.
\newblock {\em Mathematische annalen}, 261:515--534, 1982.

\bibitem[LN19]{DBLP:conf/issac/LiN19}
Jianwei Li and Phong~Q. Nguyen.
\newblock Computing a lattice basis revisited.
\newblock In {\em {ISSAC} 2019}, pages 275--282. {ACM}, 2019.

\bibitem[LP19]{DBLP:conf/issac/LiuP19}
Renzhang Liu and Yanbin Pan.
\newblock Computing hermite normal form faster via solving system of linear equations.
\newblock In {\em {ISSAC} 2019, Beijing, China}, pages 283--290. {ACM}, 2019.

\bibitem[LS22]{DBLP:conf/issac/LiS22}
Haomin Li and Arne Storjohann.
\newblock Computing a basis for an integer lattice: {A} special case.
\newblock In {\em {ISSAC} '22}, pages 303--310. {ACM}, 2022.

\bibitem[MG02]{DBLP:books/daglib/0018102}
Daniele Micciancio and Shafi Goldwasser.
\newblock {\em Complexity of lattice problems - a cryptograhic perspective}, volume 671 of {\em The Kluwer international series in engineering and computer science}.
\newblock Springer, 2002.

\bibitem[NS16]{DBLP:conf/issac/NeumaierS16}
Arnold Neumaier and Damien Stehl{\'{e}}.
\newblock Faster {LLL}-type reduction of lattice bases.
\newblock In {\em {ISSAC}}, pages 373--380. {ACM}, 2016.

\bibitem[NSV11]{DBLP:conf/stoc/NovocinSV11}
Andrew Novocin, Damien Stehl{\'{e}}, and Gilles Villard.
\newblock An {LLL}-reduction algorithm with quasi-linear time complexity: extended abstract.
\newblock In {\em {STOC}}, pages 403--412. {ACM}, 2011.

\bibitem[Poh87]{DBLP:journals/jsc/Pohst87}
Michael Pohst.
\newblock A modification of the {LLL} reduction algorithm.
\newblock {\em J. Symb. Comput.}, 4(1):123--127, 1987.

\bibitem[PS10]{pernet2010fast}
Cl{\'e}ment Pernet and William Stein.
\newblock Fast computation of hermite normal forms of random integer matrices.
\newblock {\em Journal of Number Theory}, 130(7):1675--1683, 2010.

\bibitem[Sch71]{DBLP:journals/acta/Schonhage71}
Arnold Sch{\"{o}}nhage.
\newblock {Schnelle Berechnung von Kettenbruchentwicklungen}.
\newblock {\em Acta Informatica}, 1:139--144, 1971.

\bibitem[SL96]{DBLP:conf/issac/StorjohannL96}
Arne Storjohann and George Labahn.
\newblock Asymptotically fast computation of hermite normal forms of integer matrices.
\newblock In {\em {ISSAC} '96}, pages 259--266. {ACM}, 1996.

\bibitem[WXXZ24]{williams2024new_matrixmultiplication}
Virginia~Vassilevska Williams, Yinzhan Xu, Zixuan Xu, and Renfei Zhou.
\newblock New bounds for matrix multiplication: from alpha to omega.
\newblock In {\em Proceedings of the 2024 Annual ACM-SIAM Symposium on Discrete Algorithms (SODA)}, pages 3792--3835. SIAM, 2024.

\end{thebibliography}

\begin{acronym}
\acro{hnf}[HNF]{Hermite normal form}
\acro{snf}[SNF]{Smith normal form}
\acro{de}[DE]{Diophantine Equations}
\acro{lbr}[LBR]{Lattice Basis computation}
\end{acronym}

\end{document}